\documentclass[twoside,11pt]{article}
\usepackage{jmlr2e}
\usepackage{lastpage}
\usepackage{amsmath,amsfonts,bbm,graphics,a4wide,rotating,stmaryrd,url}

\jmlrheading{23}{2022}{1-\pageref{LastPage}}{5/21}{6/22}{21-0577}{Anna Bonnet, Claire Lacour, Franck Picard, Vincent Rivoirard}
\ShortHeadings{Uniform deconvolution for Poisson Point Processes}{Bonnet, Lacour, Picard, Rivoirard}

\firstpageno{1}

\usepackage[normalem]{ulem}
\usepackage[usenames]{color}
\definecolor{plum}{rgb}{.6,0,.6}
\definecolor{forest}{rgb}{0,.7,0}
\definecolor{midnight}{rgb}{0,0,.7}

\def\Var{\mathrm{Var}}

\DeclareMathOperator*{\argmin}{argmin}

\newtheorem{Hyp}{Assumption}
\newtheorem{Cor}{Corollary}
\renewenvironment{proof}{\noindent{\bf Proof.}}{\hfill$\square$\par\noindent}



\newcommand{\E}{\ensuremath{\mathbb{E}}}

\renewcommand{\P}{\ensuremath{\mathbb{P}}}

\newcommand{\R}{\ensuremath{\mathbb{R}}}

\renewcommand{\L}[1]{\ensuremath{\mathbb{L}_{#1}}}
\newcommand{\e}{\ensuremath{\varepsilon}}

\newcommand{\1}{{\bf 1}}

\renewcommand{\H}{\mathcal{H}}
\newcommand{\M}{\mathcal{M}}

\definecolor{darkred}{rgb}{.8,.0,.035}

\numberwithin{equation}{section}

\begin{document}

\title{Uniform deconvolution for Poisson Point Processes}

\author{\name Anna Bonnet \email anna.bonnet@upmc.fr \\
       \addr 	LPSM, Sorbonne Universit\'e, UMR CNRS 8001,  \\
       75005 Paris, France \\
       \addr Univ. Lyon, Universit\'e Lyon 1, CNRS, LBBE UMR5558, \\
       F-69622 Villeurbanne, France     	   
       \AND
       \name Claire Lacour \email claire.lacour@u-pem.fr \\
       \addr Univ Gustave Eiffel, Univ Paris Est Creteil, CNRS, LAMA UMR8050, F-77447 Marne-la-Vallée, France 
       \name Franck Picard \email franck.picard@ens-lyon.fr \\
       \addr LBMC, Univ. Lyon, ENS de Lyon, UCBL, CNRS UMR 5239, INSERM U1210, \\
       46 allée d'Italie, Site Jacques Monod, 69007 Lyon, France. \\
       \addr Univ. Lyon, Universit\'e Lyon 1, CNRS, LBBE UMR5558, \\
        F-69622 Villeurbanne, France\\     
       \name Vincent Rivoirard \email vincent.rivoirard@dauphine.fr \\
	   \addr CEREMADE, CNRS, UMR 7534, \\
	   Universit\'e Paris-Dauphine, PSL University, \\
	   75016 Paris, France\\
}

\editor{Garvesh Raskutti}
\maketitle

\begin{abstract}
	We focus on the estimation of the intensity of a Poisson process in the presence of a uniform noise. We propose a kernel-based procedure fully calibrated in theory and practice. We show that our adaptive estimator is optimal from the oracle and minimax points of view, and provide new lower bounds when the intensity belongs to a Sobolev ball.  By developing the Goldenshluger-Lepski methodology in the case of deconvolution for Poisson processes, we propose an optimal data-driven selection of the kernel bandwidth. 
	Our method is illustrated on the spatial distribution of replication origins and sequence motifs along the human genome.
\end{abstract}


\section{Introduction}

Inverse problems for Poisson point processes have focused much attention in the statistical literature over the last years, mainly because the estimation of a Poisson process intensity in the presence of additive noise is encountered in many practical situations like tomography, microscopy, high energy physics. Our work is motivated by an original application field in high throughput biology, that has been revolutionized by the development of high throughput sequencing. The applications of such technologies are many, and we focus on the particular cases where sequencing allows the fine mapping of genomic features along the genome, like transcription factors. The spatial distribution of these features can be modeled by a Poisson process with unknown intensity. Unfortunately, detections are prone to some errors, which produces data in the form of genomic intervals whose width is linked to the precision of detection. Since the exact position of the peak is unknown within the interval (and not necessarily positioned at the center on average), the appropriate error distribution is uniform, the level of noise being given by the width of the intervals. Another example is provided when studying the spatial distribution of sequence motifs along the genome. A sequence motif is a pattern of nucleotides that is widespread along the genome, with potentially unknown function, but whose frequent occurrence suggests some implication in biological pathways. G-quadruplexes motifs for instance are made of guanine (G) repeats in tetrads that form particular 3D structures whose biologically function is currently unknown  \citep{CMB15}. However their implication in replication initiation has now been demonstrated \citep{PCA14} among other biological functions. These motifs are $\sim 25-30$ nucleotide long, and when studying their spatial distribution, their occurrence can be modelled by a Poisson process, and the uniform error model recalls that the data are in the form of intervals, without any reference occurrence point within the interval. Hence the spatial distribution of these motifs should be deconvoluted from this uniform error.

In the 2000s, several wavelet methods have been proposed for Poisson intensity estimation from indirect data \citep{AntoniadisBigot}, as well as B-splines and empirical Bayes estimation \citep{kuuselapanaretos}. Other authors turned to variational regularization: see the survey of \cite{hohage2016}, which also contains examples of applications and reconstruction algorithms. From a more theoretical perspective, \cite{kroll} studied the estimation of the intensity function of a Poisson process from noisy observations in a circular model. His estimator is based on Fourier series and is not appropriate for uniform noise (whose Fourier coefficients are zero except the first).

The specificity of uniform noise has rather been studied in the context of density deconvolution. In this case also, classical methods based on the Fourier transform do not work either in the case of a noise with vanishing characteristic function \citep{meister2009}. Nevertheless several corrected Fourier approaches were  introduced \citep{HallRGR01, HallMeister07,Meister08,FeuervergerKimSun08}. In this line, the work of \cite{DelaigleMeister11} is particularly interesting, even if it is limited to a density to estimate with finite left endpoint. In a recent work, \cite{BelomestnyGoldenshluger} have shown that the Laplace transform can perform deconvolution for general measurement errors. Another approach consists in using Tikhonov regularization for the convolution operator \citep{carrasco2011spectral,trong2014tikhonov}. In the specific case of uniform noise (also called boxcar deconvolution), it is  possible to use \textit{ad hoc} kernel methods \citep{GroeneboomJongbloed03,vanEs11}. In this context of non-parametric estimation, each method depends on a regularization parameter (such as  a  resolution level, a regularization parameter or a bandwidth), and only a good choice of this parameter allows to achieve an optimal reconstruction. This parameter selection is often named adaptation since the point is to adapt the parameter to the features of the target density. The above cited works (except \cite{DelaigleMeister11}) do not address this  adaptation issue or only from a practical point of view, although this is central both from the practical and theoretical points of views. 

We propose a kernel estimator to estimate the intensity of a Poisson process in the presence of a uniform noise.  We provide theoretical guarantees of its performance by deriving the minimax rates of convergence of the integrated squared risk for an intensity belonging to a Sobolev ball. To ensure the optimality of our procedure, we establish new lower bounds on this smoothness space. Then we provide an adaptive procedure for bandwidth selection using the Goldenshluger-Lepski methodology, and we show its optimality in the oracle and minimax frameworks. From the practical point of view we tune the method based on simulations to determine a consensus value for the hyperparameter.
 The empirical performance of our estimator is then studied by simulations and competed with a deconvolution method based on Gaussian errors. Finally we provide an illustration of our procedure on experimental data in Genomics, where the purpose is to study the spatial repartition of replication starting points and sequence motifs along chromosomes in humans \citep{PCA14}. The code is available at \url{https://github.com/AnnaBonnet/PoissonDeconvolution}.

\section{Uniform deconvolution model}\label{sec:model}
We consider $(X_i)_i$, the realization of a Poisson Process on $\R$, denoted by $N^X$, with $N_+:=N^X(\R)$ the number of occurrences. The uniform convolution model consists in observing $(Y_i)_i$, occurrences of a Poisson process $N^Y$, a noisy version of $N^X$ corrupted by a uniform noise, such that:
\begin{equation}\label{model}
	\forall\, i  \in \{1,\ldots, N_+\}, \quad Y_i=X_i+\e_i, \quad  \e_i \sim \mathcal{U}[-a;a],
\end{equation}
where $a$, assumed to be known, is fixed. The errors $(\e_i)_i$ are supposed mutually independent, and independent of $(X_i)_i$. Then, we denote by $\lambda_X$ the intensity function of $N^X$, and $\mu_X$ its mean measure assumed to satisfy $\mu_X(\R)<\infty$, so that  $$d\mu_X(x)=\lambda_X(x)dx,\quad x\in\R.$$
Note that $N_+\sim {\mathcal P}(\mu_X(\R)),$ 
where ${\mathcal P}(\theta)$ denotes the Poisson distribution with parameter $\theta$.
Then we further consider that observing $N^X$ with intensity $\lambda_X$ is equivalent to observing $n$ i.i.d. Poisson processes with common intensity $f_X$, with $\lambda_X= n \times f_X$. This specification will be convenient to adopt an asymptotic perspective. As for $\lambda_Y$, the intensity of $N^Y$, it can easily be shown that
$$n^{-1}\times\lambda_Y=n^{-1}\times(\lambda_X \star f_\e)= f_X\star f_\e=: f_Y,$$ where $f_\e$ stands for the density of the uniform distribution. The goal of the deconvolution method is to estimate $f_X$, based on the observation of $N^Y$ on a compact interval $[0,T]$ for some fixed positive real number $T$. In the following, we provide an optimal estimator of $f_X$ in the oracle and minimax settings. Minimax rates of convergence will be studied in the asymptotic perspective $n\to +\infty$ and parameters $a$ and $T$  will be viewed as constants. Furthermore, $\|f_X\|_1$, the $\L{1}$-norm will be assumed to be larger than an absolute constant, denoted by $r$.

\paragraph{Notations.} We denote by $\|\cdot\|_{2,T}$, $\|\cdot\|_{1,T}$ and $\|\cdot\|_{\infty,T}$ the $\L{2}$, $\L{1}$ and sup-norm on $[0;T]$, and $\|\cdot\|_2$, $\|\cdot\|_1$, and $\|\cdot\|_{\infty}$ their analog on \R. Notation $\lesssim$ means that the inequality is satisfied up to a constant and $a_n=o(b_n)$ means that the ratio $a_n/b_n$ goes to 0 when $n$ goes to $+\infty$. Finally, $\llbracket a;b \rrbracket$ denotes the set of integers larger or equal to $a$ and smaller or equal to~$b$.
\section{Estimation procedure}\label{sec:estim}

\subsection{Deconvolution with kernel estimator}\label{sec:heuristics}
{To estimate $f_X$ based on observations of $N^Y$, we introduce a kernel estimator which is based on the following heuristic arguments inspired from \cite{vanEs11} who considered the setting of uniform deconvolution for density estimation. We observe that $f_Y$ can be expressed by using the cumulative distribution of the $X_i$'s:
$$F_X(x):=\int_{-\infty}^x f_X(u)du\leq \|f_X\|_1,\quad x\in\R.$$
Indeed, for $x\in\R$,
\begin{align}\label{eq:f_Y_to_F_X}
	f_Y(x) & = \int_\R f_X(x-u)f_\e(u)du\nonumber\\
	& = \frac{1}{2a}\int_{-a}^af_X(x-u)du\nonumber\\
	&= \frac{1}{2a}\bigg( F_X(x+a)-F_X(x-a) \bigg),
\end{align}
from which we deduce:
$$
F_X(x)=2a\sum_{k=0}^{+\infty}f_Y\bigg(x-(2k+1)a \bigg),\quad x\in\R.
$$
Then, from heuristic arguments, we get
\begin{equation}\label{fY'}
	f_X(x)=2a\sum_{k=0}^{+\infty}f'_Y\bigg(x-(2k+1)a \bigg),
\end{equation}
which provides a natural form of our kernel estimate $\widehat{f}_Y$. Note that differentiability of $f_Y$ is not  assumed in the following.
Indeed,  we consider the kernel estimator of $f_Y$ such that
$$
\widehat f_Y(x) =\frac{1}{nh}\int_{\R} K\left(\frac{x-u}{h}\right)dN^Y_u,
$$
with $dN^Y$ the point measure associated to $N^Y$, and $K$ a kernel with  bandwidth $h>0$. Setting $$K_h(x)=\frac{1}{h}K\bigg(\frac{x}{h}\bigg),$$ we can write
$$\widehat f_Y(x) =\frac{1}{n}\sum_{i=1}^{N_+}K_h\big(x-Y_i\big).
$$
Then, if $K$ is differentiable, we propose the following kernel-based estimator of $f_X$
$$\widehat f_h(x) = \frac{2a}{nh^2}\sum_{k=0}^{+\infty} \sum_{i=1}^{N_+}K'\left(\frac{x-(2k+1)a-Y_i}{h}\right).$$
The proof of subsequent Lemma~\ref{lem:varsym} in Appendix shows that the expectation of $\widehat{f}_h$ is a regularization of $f_X$, as typically desired for kernel estimates since we have
\begin{equation}\label{mean}
	\E[\widehat{f}_h]= K_h \star f_X.
\end{equation}
Then,  our objective is to provide an optimal selection procedure for the parameter $h$.  
\subsection{Symmetrization of the estimator}
Our estimator is based on the inversion and differentiation of Equation~\eqref{eq:f_Y_to_F_X},  which can also be performed as follows:
\begin{align*}
	2a\sum_{k=0}^{+\infty}f_Y(x+(2k+1)a)=\sum_{k=0}^{+\infty}
	\bigg( F_X(x+2(k+1)a)-F_X(x+2ka) \bigg)=\|f_X\|_1 -F_X(x)
\end{align*}
and differentiated to obtain:
$$f_X(x)=-2a\sum_{k=0}^{+\infty}f'_Y \bigg( x+(2k+1)a \bigg),$$
which leads to another estimator 
$$\check{f}_h(x)=-\frac{2a}{nh^2}\sum_{k=0}^{+\infty}\sum_{i=1}^{N_+}K'\left(\frac{x+(2k+1)a-Y_i}{h}\right).$$
In the framework of uniform deconvolution for densities, \cite{vanEs11} proposes to use $\alpha \hat f_h(x)+(1-\alpha) \check  f_h(x)$, a convex combination of $\hat f_h$ and $ \check  f_h$, as a combined estimator, to benefit from the small variance of $\check{f}_h(x)$ and $\hat{f}_h(x)$ for large and small values of $x$ respectively. Unfortunately the combination that minimizes the asymptotic variance of the combined estimator is achieved for $\alpha=1-F_X(x)$, and thus depends on an unknown quantity. \cite{vanEs11} suggested to use a plug-in estimator, but to avoid supplementary technicalities, we finally consider the following symmetric kernel-based estimator: 
\begin{equation}\label{def:final}
	\widetilde{f}_h(x):=\frac{1}{2} \bigg( \hat f_h(x) + \check  f_h(x) \bigg) =\frac{a}{nh^2}\sum_{k=-\infty}^{+\infty}s_k\sum_{i=1}^{N_+}K'\left(\frac{x-(2k+1)a-Y_i}{h}\right),
\end{equation}
with $s_k=1$ if $k\geq 0$ and $s_k=-1$ if $k<0$. Then it is shown in Lemma~\ref{lem:varsym} that
$$\E[\widetilde{f}_h]=(K_h \star f_X).$$
\subsection{Risk of the kernel-based estimator}
Our objective is to provide a selection procedure to select a bandwidth $\widehat{h}$, that only depends on the data, so that the $\L{2}$-risk of $\widetilde{f}_{\widehat{h}}$ is smaller than the risk of the best kernel estimate (up to a constant), namely
$$\E\left[\|\widetilde f_{\widehat{h}}-f_X\|_{2,T}^2\right]\lesssim \inf_{h\in{\mathcal H}} \E\left[\|\widetilde f_h-f_X\|_{2,T}^2\right].$$
Our procedure is based on the bias-variance trade-off of the risk of any estimate $\widetilde f_h$
\begin{align}\label{BV}
	\E\Big[\|\widetilde f_h- f_X]\|_{2,T}^2\Big]&=\|\E[\widetilde f_h]-f_X\|_{2,T}^2+ \E\Big[\|\widetilde f_h-\E[\widetilde f_h]\|_{2,T}^2\Big]\nonumber\\
	&=:B_h^2+v_h.
\end{align}
Then we use the following mild assumption:
\begin{Hyp}\label{Hyp1}
	The kernel $K$ is supported on the compact interval $[-A,A]$, with $A\in\R_+^*$ and $K$ is  differentiable on $[-A,A]$.
\end{Hyp}
Then the variance of the estimator is such that:
\begin{lemma}\label{lem:varsym}
	For any $h\in\mathcal H$ and any $x\in [0,T]$, we have
	$$\E[\widetilde f_h(x)]=(K_h\star f_X)(x).$$
	Under Assumption~\ref{Hyp1} and if $h$ is small enough so that $Ah\leq a$,
	then
	$$v_h:=\E\Big[\|\widetilde f_h-\E[\widetilde f_h]\|_{2,T}^2\Big] = \frac{aT\|f_X\|_1 \|K'\|_2^2}{2nh^3}.$$
\end{lemma}
The expectation of $\widetilde f_h$ has the expected expression (derived from \eqref{mean}), but Lemma~\ref{lem:varsym} also provides the exact expression of the variance term $v_h$ of our very specific estimate. Since our framework is an inverse problem, this variance does not reach the classical $(nh)^{-1}$ bound. Moreover, $v_h$ depends linearly on $\|f_X\|_1$ but also on $T$, which means that the estimation of $f_X$ has to be performed on the compact interval $[0,T]$, for $v_h$ to be finite. This requirement is due to Expression~\eqref{def:final} of our estimate that shows that for any $x\in\R$, $\widetilde f_h(x)$ is different from 0 almost surely. This dependence of $v_h$ on $T$ is a direct consequence of our strategy not to make any assumption on the support of $f_X$, that can be unknown or non-compact. Of course, if the support of $f_X$ was known to be compact, like $[0,1]$, then we would force $\widetilde f_h$ to be null outside $[0,1]$ (for instance by removing large values of $|k|$ in the sum of \eqref{def:final}), and estimation would be performed on the set $[0,1]$. Actually, estimating a non-compactly supported Poisson intensity on the whole real line leads to deterioration of classical non-parametric rates in general \citep{RBR}.
\subsection{Bandwidth selection}\label{sec:bandwidth}
The objective of our procedure is to choose the bandwidth $h$, based on the Goldenshluger-Lepski methodology \citep{GL13}. First, we introduce a finite set $\H$ of bandwidths such that for any $h\in \H$, $h \leq a/A,$ which is in line with assumptions of Lemma~\ref{lem:varsym}. Then, for two bandwidths $t$ and $h$, we also define 
$$
\widetilde f_{h,t}:=K_h\star \widetilde f_t,
$$
a twice regularized estimator, that satisfies the following property (see Lemma~\ref{lem:inversion} in Appendix):
$$\widetilde f_{h,t}=\widetilde f_{t,h}.$$
Now we select the bandwidth as follows:
\begin{equation}
	\label{bv}\widehat h:=\argmin_{h\in \mathcal H}\left\{{\mathcal A}(h)+\frac{c\sqrt{N_+}}{nh^{3/2}}\right\},
\end{equation}
where
\begin{equation}\label{c}
	c=(1+\eta)(1+ \|K\|_1) \|K'\|_2\sqrt{\frac{aT}{2}}
\end{equation}
for some $\eta>-1$ and 
$${\mathcal A}(h):=\max_{t\in \mathcal H}\left\{\|\widetilde f_{h,t}-\widetilde f_t\|_{2,T}-\frac{c\sqrt{N_+}}{nt^{3/2}}\right\}_+.$$
Finally, we estimate $f_X$ with 
\begin{equation}\label{def:est}
	\widetilde  f=\widetilde f_{\widehat h}.
\end{equation}
Note that ${\mathcal A}(h)$ is an estimation of the bias term $B_h$ of the estimator $\widetilde f_h$. Indeed,
$$B_h:=\|\E[\widetilde f_h]-f_X\|_{2,T}=\|K_h\star f_X-f_X\|_{2,T}$$
and we replace the unknown function $f_X$ with the kernel estimate $\widetilde f_t$. The term $\frac{c\sqrt{N_+}}{nt^{3/2}}$ in ${\mathcal A}(h)$ controls the fluctuations of $\|\widetilde f_{h,t}-\widetilde f_t\|_{2,T}$. Finally, since $\E[N_+]=n\|f_X\|_1$, \eqref{bv} mimics the bias-variance  trade-off \eqref{BV} (up to the squares). In order to fully define the estimation procedure, it remains to choose the set of bandwidths ${\mathcal H}$. This is specified in Section~\ref{sec:theoretical}.
\section{Theoretical results}\label{sec:theoretical}
\subsection{Oracle approach}
The oracle setting allows us to prove that the bandwidth selection procedure described in Section~\ref{sec:bandwidth} is (nearly) optimal among all kernel estimates. Indeed, we obtain the following result.
\begin{theorem}\label{theo:GL}
	Suppose that Assumption~\ref{Hyp1} is verified.
	We take $\eta>0$ and we consider the estimate $\widetilde f$ such that the finite set of bandwidths $\mathcal H$ satisfies
	$$\min \mathcal H=1/(\delta n^{\frac13})\quad \mbox{and}\quad \max\mathcal H=o(1)$$
	for some constant $\delta>0$. Then, for $n$ large enough,
	\begin{equation}\label{UB}
		\E\left[\|\widetilde f_{\widehat h}-f_X\|_{2,T}^2\right]\leq C_1\inf_{h\in{\mathcal H}} \E\left[\|\widetilde f_h-f_X\|_{2,T}^2\right]+\frac{C_2}{n},
	\end{equation}
	where $C_1=2+24(1+\eta)^2(1+\|K\|_1)^2$ and $C_2$ is a constant depending on $\|f_X\|_1$, $T$, $a$, $\delta,$ $\eta$ and $K$.
\end{theorem}
The proof of this result can be found in Section~\ref{sec:preuvebornesup}, where the expression of $C_2$ is provided (see Equation~\eqref{C2}).
\begin{remark}
Note that condition $h\geq\delta^{-1} n^{-1/3}$ is equivalent to $\sup_n\sup_{h\in{\mathcal H}}v_h <\infty$
\end{remark}
\begin{remark}
Equation~\eqref{C2} provides the explicit dependence of $C_2$ on $a,$ $\delta,$ $\eta,$ $T$ and $K$, showing that the kernel $K$ has to be chosen such that $\|K\|_1,$ $\|K'\|_\infty$, $\|K\|_2,$ $\|K'\|_1$ and $\|K'\|_2$ are as small as possible. Nevertheless, in the minimax approach of Section~\ref{sec:minimax}, the kernel has to satisfy some constraints (see Assumption~\ref{Hyp2}). The parameter $a$ is present in the remainder term of the risk bound in the following way $C_2/n=(k_0+k_1a+k_2a^2)/n$. Thus the larger $a$ the worse the bound, this is expected since $a$ measures the noise level.
\end{remark}

Theorem~\ref{theo:GL} shows that our procedure achieves nice performance: Up to the constant $C_1$ and the negligible term $C_2/n$ that goes to 0 quickly, our estimate has the smallest risk among all kernel rules under mild conditions on the set of bandwidths $\mathcal H.$
\subsection{Minimax approach}\label{sec:minimax}
The minimax approach is a framework that shows the optimality of an estimate among all possible estimates. For this purpose, we consider a class of functional spaces for $f_X$, then we derive the minimax risk associated with each functional space and show that our estimator achieves this rate. Here, we consider the class of Sobolev balls that can be defined, for instance, through the Fourier transform of $\L{2}$-functions: Given $\beta>0$, $L>0$, $b>0$ and $r>0$, consider the following subset of the Sobolev ball of smoothness $\beta$ and radius $L$ 
$${\mathcal S}^\beta(L,r,b):=\left\{g\in\L{2}:\quad \int_{-\infty}^{+\infty}|g^*(\xi)|^2(\xi^2+1)^\beta d\xi\leq L^2,\ r\leq\|g\|_1\leq bL \right\},$$

where $g^*(\xi):=\int e^{ix\xi}g(x)dx$ is the Fourier transform of $g$. Observe that the classical Sobolev space corresponds to the case $r=0$ and $b=+\infty$; $r$ and $b$ will be viewed as constants in the sequel. In the Poisson setting, the ${\mathbb L}_1$-norm of the Poisson intensity is not fixed but it of course plays a key role in rates. 
Given $L$, the radius of the Sobolev ball containing $g$, the ${\mathbb L}_1$-norm of $g$ scales in $L$.
We finally introduce the lower bound $\|g\|_1\geq r$ with $r>0$ to avoid the asymptotic setting where the Poisson intensity goes to 0. 

From a statistical perspective, the minimax rate associated with the space ${\mathcal S}^\beta(L,r,b)$ is
$${\mathcal R}_n(\beta,L):=\inf_{Z_n}\sup_{f_X\in{\mathcal S}^\beta(L,r,b)}
\E\left[\|Z_n-f_X\|_{2,T}^2\right],$$
where the infimum is taken over all estimators $Z_n$ of $f_X$ based on the observations $(Y_i)_{i=1,\ldots,N_+}.$ In the notation, we drop the dependence of the risk on $r$ and $b$ since we are only interested in the dependence on $n$, $\beta$ and $L$. We first derive a lower bound for the minimax risk.

\begin{theorem}\label{t2} We assume that $rL^{-1}\leq \pi/(2\mathfrak c_\beta)\leq b$, where $\mathfrak c_\beta$ is defined in \eqref{cbeta}. There exists a positive constant $C_3$ only depending on $\beta,$  $a$ and $T$ such that, if $n$ is larger than some $n_0$ only depending on $r$ and $T$,
	\begin{equation}\label{rate:LB}
		{\mathcal R}_n(\beta,L) \geq C_3 \left[L^{\frac{2\beta+6}{2\beta+3}}n^{-\frac{2\beta}{2\beta+3}}+Ln^{-1}\right].
	\end{equation}
\end{theorem}
Theorem~\ref{t2} is proved in Section~\ref{prooflowerbound}. To the best of our knowledge, because of the second term $Ln^{-1}$, the rate  established in \eqref{rate:LB} is new. Of course, if $L$ is bounded then the second term is negligible with respect to the first one when $n\to+\infty$.
The rate $n^{-\frac{2\beta}{2\beta+3}}$ is slower than the classical non-parametric rate $n^{-\frac{2\beta}{2\beta+1}}.$ It is the expected rate since our deconvolution problem corresponds to an inverse problem of order 1, meaning that $f_\e^*$, the characteristic function of the noise, satisfies
\begin{equation*}\label{cond:TF}
	|f_\e^*(\xi)|=O(\xi^{-1}) 
	\text{ as }\xi\to \infty.
\end{equation*}
Note that the analog of the previous lower bound has been established in the  density deconvolution context, first by \cite{fan1993}, but with supplementary assumption $|{f_\e^*}'(\xi) |=O(\xi^{-2})$ which is not satisfied in our case of uniform noise (see Equation~\eqref{majo-derivee}). Our proof is rather inspired by the work  of \cite{meister2009}, 
but we face here a Poisson inverse problem, and we have to control the ${\mathbb L}_2$-norm on $[0, T]$ instead of $\R.$ 
Furthermore, Theorem~2.14 of \cite{meister2009} only holds for $\beta>1/2$. Consequently, we use different techniques to establish Theorem~\ref{t2}, which are based on wavelet decompositions of the signal. Specifically, we use the Meyer wavelets of order 2. 

We now show that the rate achieved by our estimate $\widetilde f$ corresponds to the lower bound~\eqref{rate:LB}, up to a constant. We have the following corollary, easily derived from Theorem~\ref{theo:GL}, and based on the following assumption.
\begin{Hyp}\label{Hyp2}
	The kernel $K$ is of order $\ell=\lfloor\beta\rfloor$, meaning that the functions
	$x\mapsto x^jK(x)$, $j=0,1,\ldots,\ell$ are integrable and satisfy
	$$\int K(x)dx=1,\quad \int x^jK(x)dx=0, \quad j=1,\ldots,\ell.$$
\end{Hyp}
\begin{remark}
	See Proposition~1.3 of \cite{tsybakov} or Section~3.2 of \cite{GL14} for the construction of kernels satisfying Assumptions~\ref{Hyp1} and \ref{Hyp2}.
\end{remark}

\begin{Cor}\label{cor} Suppose that Assumptions~\ref{Hyp1} and \ref{Hyp2} are satisfied. We take $\eta>0$ and we consider the estimate $\widetilde f$ such that the set of bandwidths $\mathcal H$ is 
	$$\mathcal H=\left\{D^{-1}:\quad D\in\llbracket \log n;\delta n^{\frac13}\rrbracket\right\},$$ 
	for some constant $\delta>0$. Then, for $n$ large enough,

	$$\sup_{f_X\in{\mathcal S}^\beta(L,r,b)}\E\left[\|\widetilde f-f_X\|_{2,T}^2\right]\leq  C_4 \left[L^{\frac{2\beta+6}{2\beta+3}}n^{-\frac{2\beta}{2\beta+3}}+Ln^{-1}\right],$$
	where $C_4$ only depends on $\delta$, $\eta$, $K,$ $\beta,$ $r, b$, $a$ and $T$.

\end{Cor}

Corollary~\ref{cor}, proved in Section~\ref{sec:proofcor}, shows that our estimator is adaptive minimax, i.e. it achieves the best possible rate (up to a constant) and the bandwidth selection does not depend on the spaces parameters $(\beta,L)$ on the whole range $\{0<\beta<\ell+1, L>0\}$, where $\ell$ is the order of the chosen kernel. We have established the optimality of our procedure.
\section{Simulation study and numerical tuning}
\label{sec:numerical}

In the following we use numerical simulations to tune the hyperparameters of our estimator and to assess the performance of our deconvolution procedure. We consider different shapes for the Poisson process intensity to challenge our estimator in different scenarii, by first generating Poisson processes on $[0,1]$ based on the Beta probability distribution function,  with $f_{\text{unisym}}= \text{Beta}(2,2)$ (unimodal symmetric),  $f_{\text{bisym}}= 0.5 \times  \text{Beta}(2,6) + 0.5 \times \text{Beta}(6,2)$ (bimodal symmetric), $f_{\text{biasym}}= 0.5 \times \text{Beta}(2,20) + 0.5 \times \text{Beta}(2,2)$ (bimodal assymmetric). We also generate Poisson processes with Laplace distribution intensity (location $5$, scale $0.5$) to consider a sharp form and a different support. In this case, we consider that $T=10$. We consider a uniform convolution model with increasing noise ($a \in \{0.05, 0.1 \}$ for Beta, $a \in \{0.5, 1, 2, 3 \}$ for Laplace) and Poisson processes with increasing number of occurrences ($n \in \{500,1000 \}$). For each set $(f_X,n,a)$, we present the median performance over 30 replicates. In order to keep a bounded variance of our estimators, we explore different values of $h$ using a grid denoted by $\mathcal{H}$, with minimum value $h_{min}=\left(aT/n \right)^{1/3}$ (see Lemma~\ref{lem:varsym} and Corollary~\ref{cor}). For Beta intensity we consider a grid $\mathcal{H}$ from $h_{min}$ to $0.5$ with steps of $0.025$ and for Laplace intensities from $h_{min}$ to 10 with steps of $0.5$.  Finally, our procedure is computed with an Epanechnikov kernel, that is $K(u)= 0.75 (1-u^2) \mathbbm{1}_{\vert u \vert \leq 1}$. Our estimator is challenged by the oracle estimator, that is the estimator $\widehat f_{h^*}$, with $h^*$  minimizing (with respect to $h$) the mean squared error $\mathbb{E}\|f_X-\widehat f_{h}\|_{2,T}^2$, with $f_X$ the true intensity. 

To assess the interest of designing a deconvolution method dedicated to uniform noise, our method is also competed with a deconvolution procedure for Gaussian noise  \citep{MR2045631}, available in the \texttt{fDKDE} \texttt{R}-package. All methods are compared to a density estimator without deconvolution calibrated by cross-validation.


%
\subsection{Hyperparameter tuning}
\label{subsec:calib_eta_fix}

Our selection procedure for parameter $h$ is based on the two-step method described in Section \ref{sec:estim}. Using this procedure in practice requires to tune the value of the hyper-parameter $\eta$ that is part of the penalty $c (\eta)$:
$$c(\eta)=(1+\eta)(1+\|K\|_1)\sqrt{\frac{aT}{2}}\|K'\|_2,$$
with $K$ the Epanechnikov kernel in our simulations. This penalty is at the core of the two-step method that consists in computing:
\begin{equation}A_\eta(h)=\max_{t\in \mathcal H}\left\{\|\widetilde f_t-\widetilde f_{h,t}\|_{2,T}-c(\eta) \frac{\sqrt{N_+}}{nt^{3/2}}\right\}_+, 
	\label{eq:A_h}
\end{equation}
followed by
\begin{equation}
	\widehat{h} =\argmin_{h\in \mathcal H}\left\{A_{\eta}(h)+ c(\eta) \frac{\sqrt{N_+}}{nh^{3/2}}   \right\}.
	\label{eq:h_hat}
\end{equation}

We propose to investigate if we could find a "universal" value of parameter $\eta$ that would be appropriate whatever the form of the intensity function. For a grid of $\eta$ in $[-1;1]$, we compare the mean squared errors (MSE) of estimators calibrated with different values of $\eta$ to the MSE achieved by the oracle estimator that achieves the smallest MSE over the grid~$\mathcal{H}$.
Figure \ref{fig:MSE_eta_fix_beta} shows that the optimal choice of $\eta$ depends on the shape of the true intensity ($\sim-0.15$ for Beta,$\sim -0.9$ for Laplace). For Beta intensities, the MSE curve is minimal and almost flat for $\eta \in [-0.4,0]$. Since in this range the MSE remains close to its oracle for Laplace intensities, we propose to choose $\eta = -0.3$ as a reasonable trade-off to obtain good performance in most settings.


\begin{figure}
	\begin{center}
		\includegraphics[scale=0.4,trim=0mm 0mm 0mm 0mm]{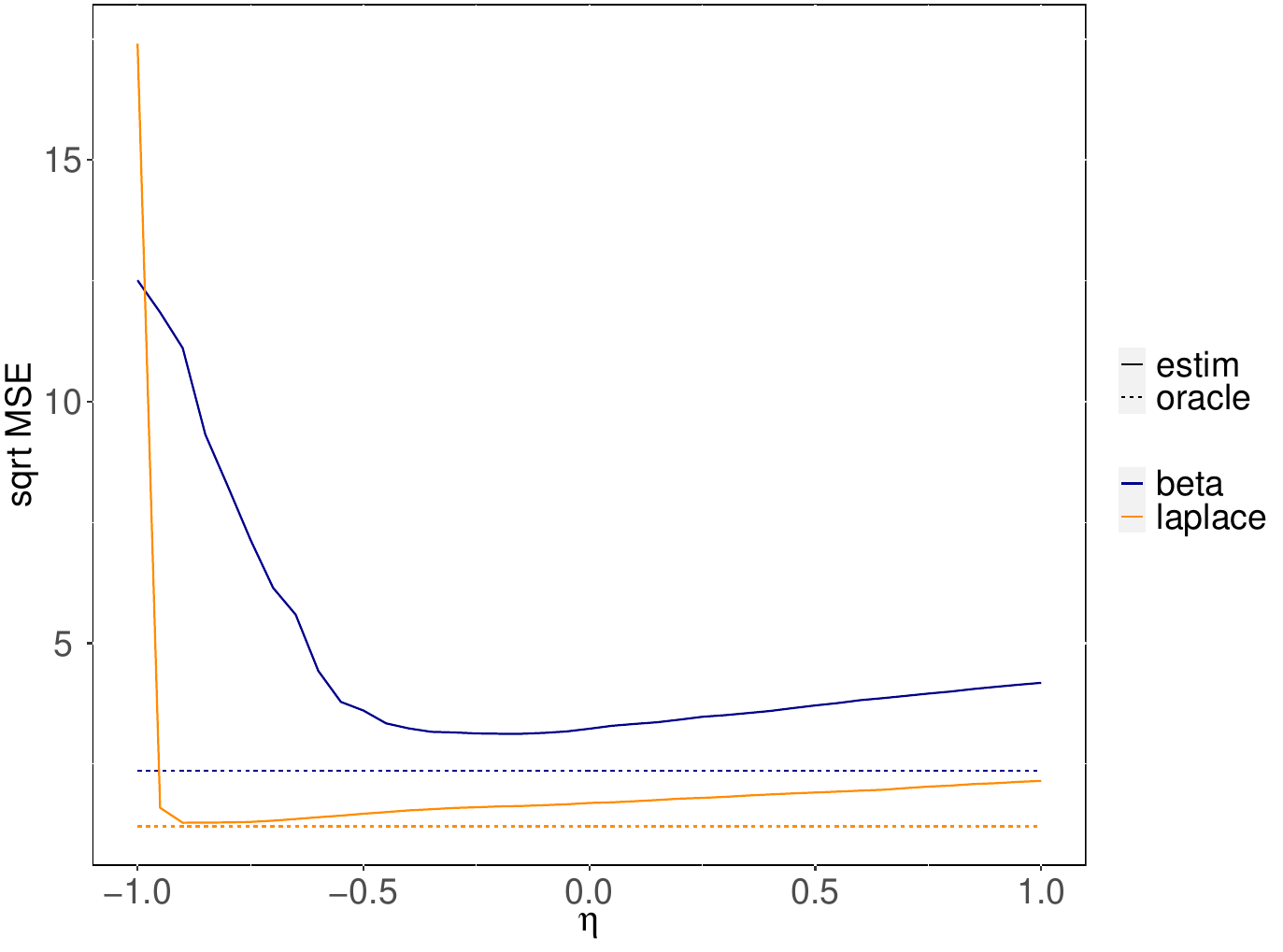}  
	\end{center}
	\caption{Mean squared error (square root) obtained on 30 simulations with our estimator for different values of $\eta$ and with an oracle estimator (dotted line). The mean squared errors are computed as mean values on 30 simulations and over all different Beta scenarii (blue line) and Laplace scenarii (orange line)}
	\label{fig:MSE_eta_fix_beta} 	
\end{figure}

\subsection{Results and comparison with other methods}

\begin{figure}
	\begin{center}  	
		\begin{tabular}{c}
			Beta intensities \\ 
			\includegraphics[scale=0.5,trim=0mm 0mm 0mm 0mm]{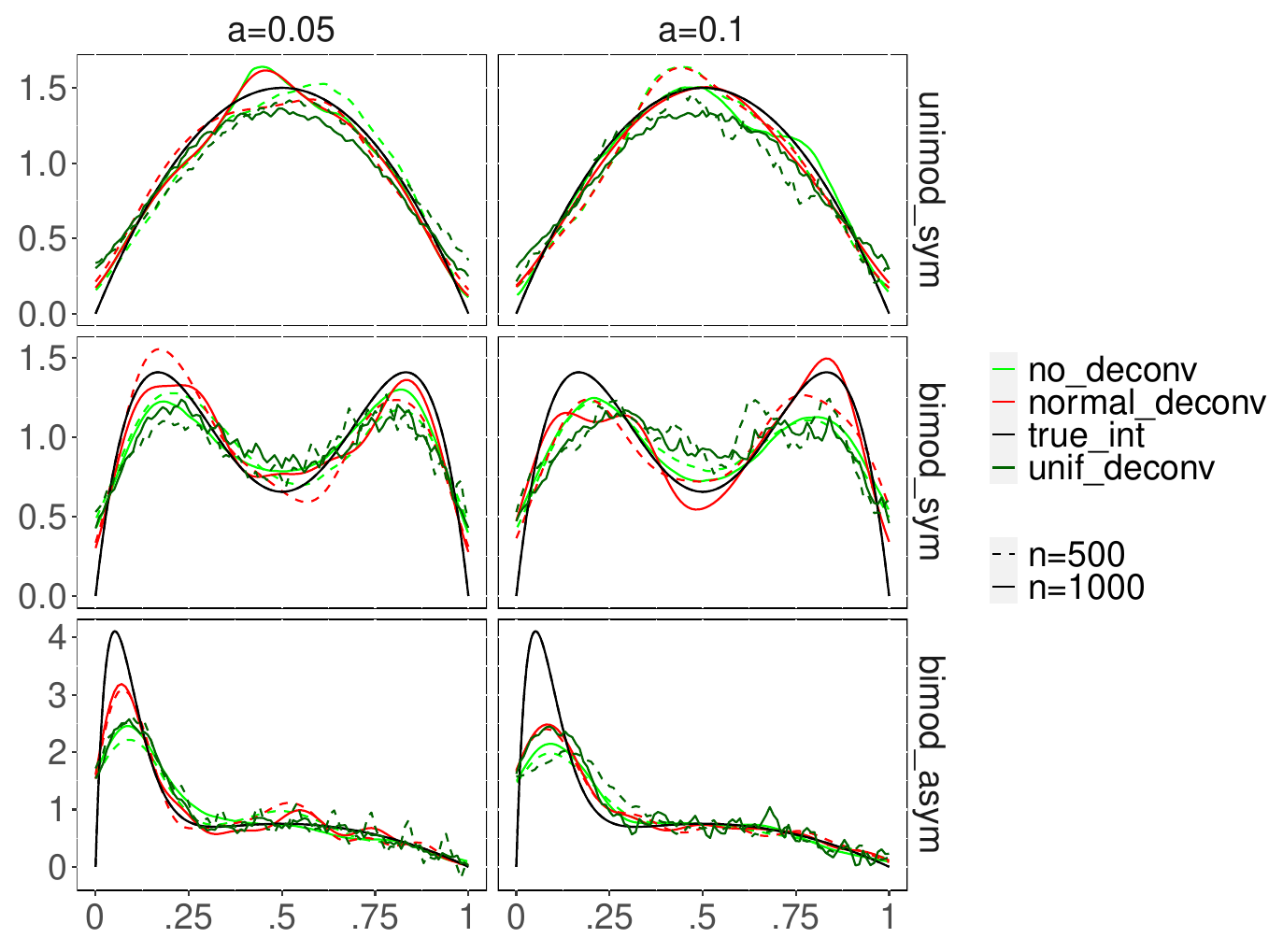} \\
			Laplace intensities \\
			\includegraphics[scale=0.5,trim=0mm 0mm 0mm 0mm]{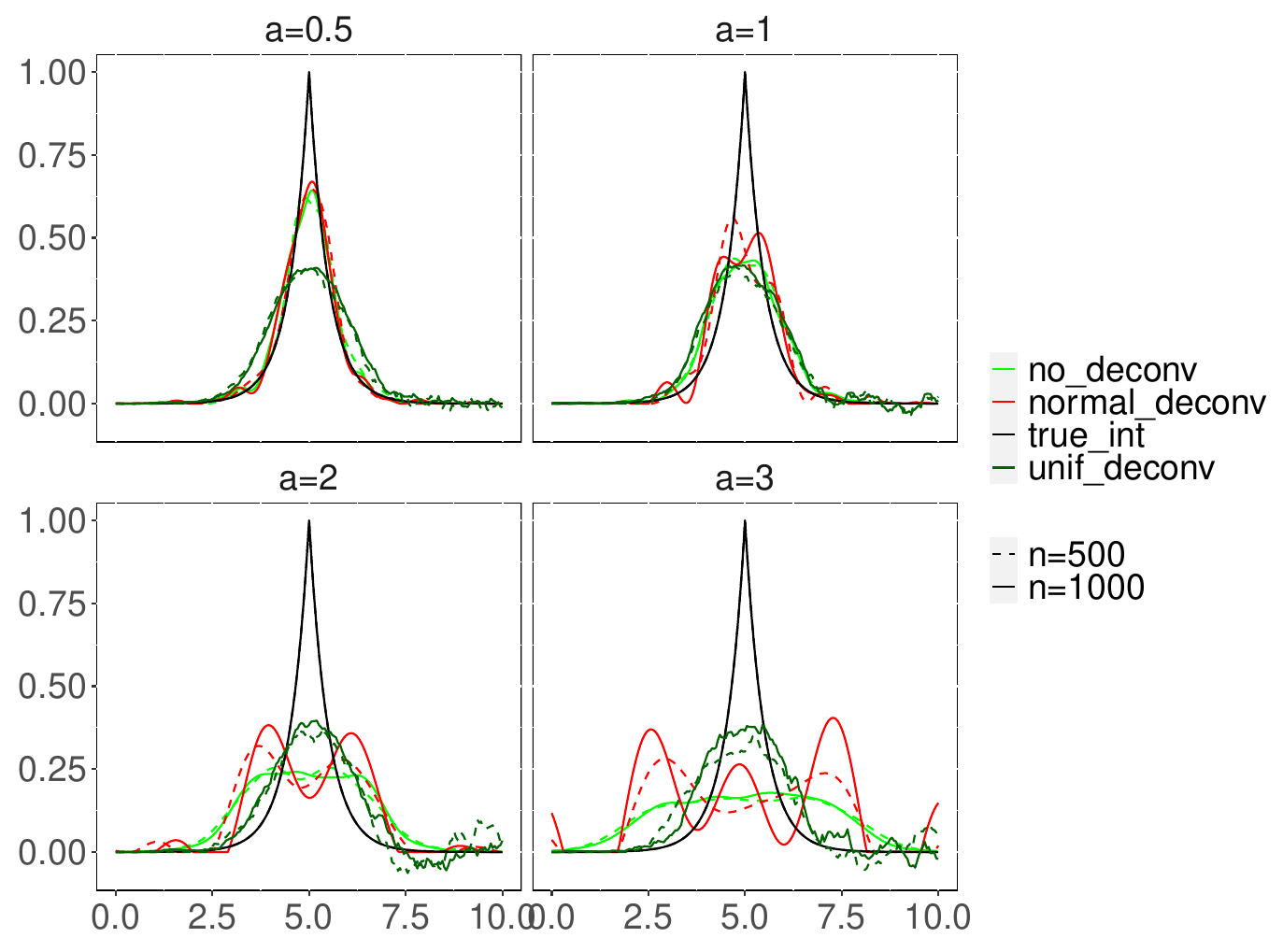}
		\end{tabular}
	\caption{Estimated intensity obtained with our estimator (\texttt{unif$\_$deconv}), with the true intensity (\texttt{true$\_$int}) and the \texttt{fDKDE} estimator (\texttt{normal$\_$deconv}) and with a density estimator without deconvolution (\texttt{no$\_$deconv}) with Epanechnikov kernel. Each reconstruction is obtained with a generated dataset that verifies the median MSE over 30 simulations.}
		\label{fig:compar_delaigle}
	\end{center}  	
\end{figure}

Figure \ref{fig:compar_delaigle} highlights two different behaviours depending on the size of the noise and the shape on the true intensity: when the noise is small ($a \leq 0.1$ for beta intensities and $a \leq 1$ for Laplace intensities), the estimator proposed by \cite{MR2045631} is very efficient. This was quite expected that the distribution of the noise would not matter when its variance is small: we see indeed that the density estimator without deconvolution also performs well in such context. However, when the noise increases and the true intensity is sharp ($a \geq 2$ for Laplace intensities), we observe major differences and our method designed for uniform noises is the only one that can provide an accurate intensity estimation. These results are confirmed with the mean-squared errors computed for each method and displayed in Figure~\ref{fig:compar_mse_delaigle}.

These results motivate the application on genomic data proposed in Section \ref{sec:appli}, where the measurement errors can be large compared to the average distance between points.

\begin{figure}
	\begin{center}  	
		\begin{tabular}{c}
			Beta intensities \\ 
			\includegraphics[scale=0.5,trim=0mm 0mm 0mm 0mm]{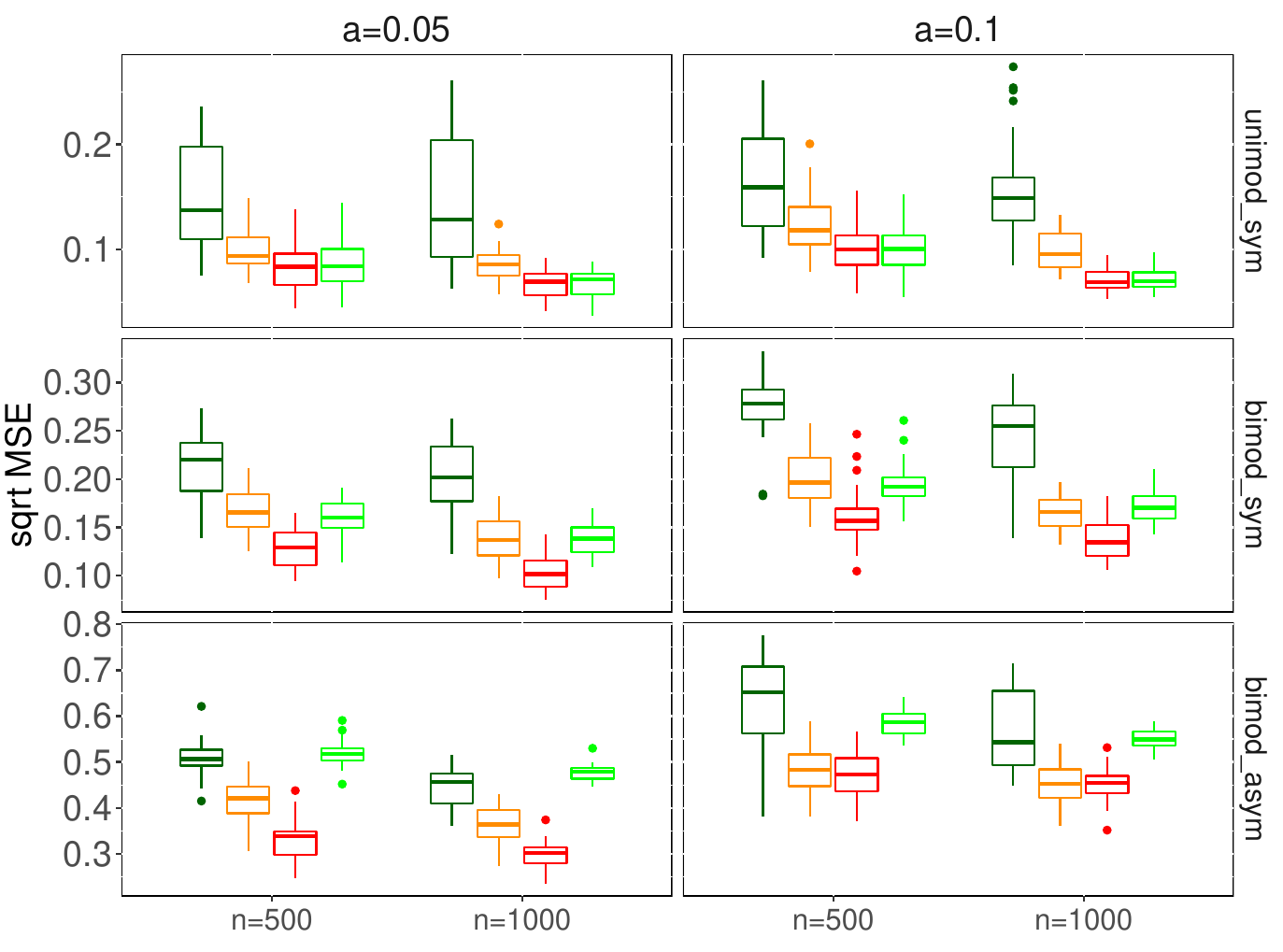} \\
			Laplace intensities \\
			\includegraphics[scale=0.5,trim=0mm 0mm 0mm 0mm]{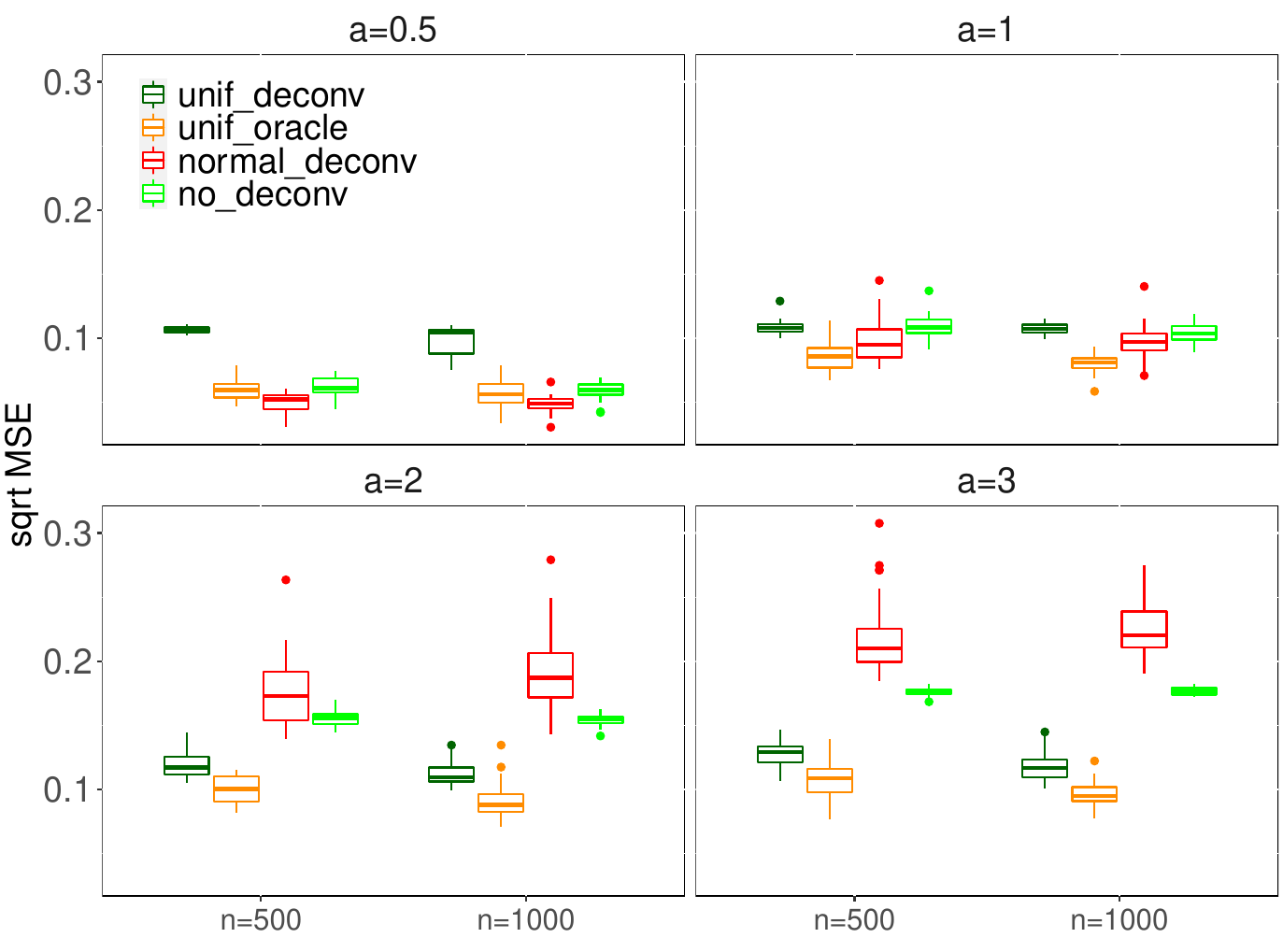}
		\end{tabular}
		\caption{Mean squared errors (square root) obtained with our estimator \texttt{unif$\_$deconv}), with the oracle estimator \texttt{unif$\_$oracle}) for uniform deconvolution, with the \texttt{fDKDE} estimator \texttt{normal$\_$deconv}) and with a density estimator without deconvolution \texttt{no$\_$deconv}) with Epanechnikov kernel. The estimations are provided for 30 simulations in each scenario. }
		\label{fig:compar_mse_delaigle}
	\end{center}  	
\end{figure}

\subsection{Computational times}

The calibration procedure for the bandwidth selection requires multiple integral computations, in particular if we use a thin grid $\mathcal{H}$. However, the computational times remains reasonable for one estimation especially when the size of the noise is not too small, as summarized in Table \ref{tab:comp_times}. The value of $a$  determines indeed the number of non-zero terms in the double sum that appears in the definition of the estimator \eqref{def:final} (the smaller $a$, the larger number of terms), which explains that the longest computational times is obtained for the smaller value of $a$ and the larger number of observations $n$. The code, implemented in \texttt{R}, is parallelized and uses the \texttt{Rcpp} package in order to reduce the computational cost.

\begin{table}
\begin{minipage}{6cm}
\quad \quad (a) Beta intensities \\
\begin{tabular}{c|c|c}
 & a=0.05 & a=0.1 \\
 \hline
n=500 & 1123 & 241 \\
 \hline
 n=1000 & 4931 & 1144
\end{tabular}
\end{minipage}
\hfill
\begin{minipage}{10cm}
\quad \quad \quad \quad (b) Laplace intensities \\
\begin{tabular}{c|c|c|c|c}
 & a=0.5 & a=1 & a=2 &a=3 \\
 \hline
n=500 &  2.79 & 0.71 & 0.31 & 0.29\\
 \hline
 n=1000 & 9.76 & 2.08 & 0.78 & 0.76
\end{tabular}
\end{minipage}
\caption{Computational times (in seconds) associated with one estimation for different values of $a$ and $n$ considered in the numerical study. The computations were run on 16 cores of a server Intel Xeon E5-4620 2.20GHz .}
\label{tab:comp_times}
\end{table}

\section{Deconvolution of Genomic data}\label{sec:appli}

Next generation sequencing technologies (NGS) have allowed the fine mapping of eukaryotes replication origins that constitute the starting points of chromosomes duplication. To maintain the stability and integrity of genomes, replication origins are under a very strict spatio-temporal control, and part of their positioning has been shown to be associated with cell differentiation \citep{PCA14}. The spatial organization has become central to better understand genomes architecture and regulation. However, the positioning of replication origins is subject to errors, since any NGS-based high-throughput mapping consists of peak-calling based on the detection of an exceptional enrichment of short reads \cite{PCA14}. Consequently, the true positions of the replication starting points are unknown, but rather inferred from genomic intervals. The spatial control of replication being very strict, the precise quantification of the density of origins along chromosomes is central, but should account for this imprecision of the mapping step. The interval shape of the data makes the uniform assumption of the noise particularly appropriate. However, other types of distributions could be considered. We compare our results to those obtained by the estimator proposed by \cite{MR2045631} and implemented in the R package \texttt{fDKDE}, which was developed to handle errors with Gaussian distribution. Both deconvolution estimators provide an intensity estimation that is less smooth than the one obtained without accounting for the error positioning. However, the estimator of \cite{MR2045631}  identifies three regions with a high density of origins while ours shows several sharp peaks which suggests the existence of clusters of origins, the location of which can be precisely identified.

The comparison shows that the Gaussian-based estimator is overly smooth regarding the underlying biological process. Indeed, replication origins are known to be organized according to the so-called replication domains that are $\leq 1$Mb on average \citep{pmid25409831}. The deconvoluted estimator based on uniform errors provides an intensity that shows peaks that are approximatively $\leq 1$Mb wide, whereas the Gaussian-based estimator shows clusters of size $\sim 3$Mb.
\begin{figure}
	\begin{center}
		\includegraphics[scale=0.5,trim=0mm 0mm 0mm 0mm]{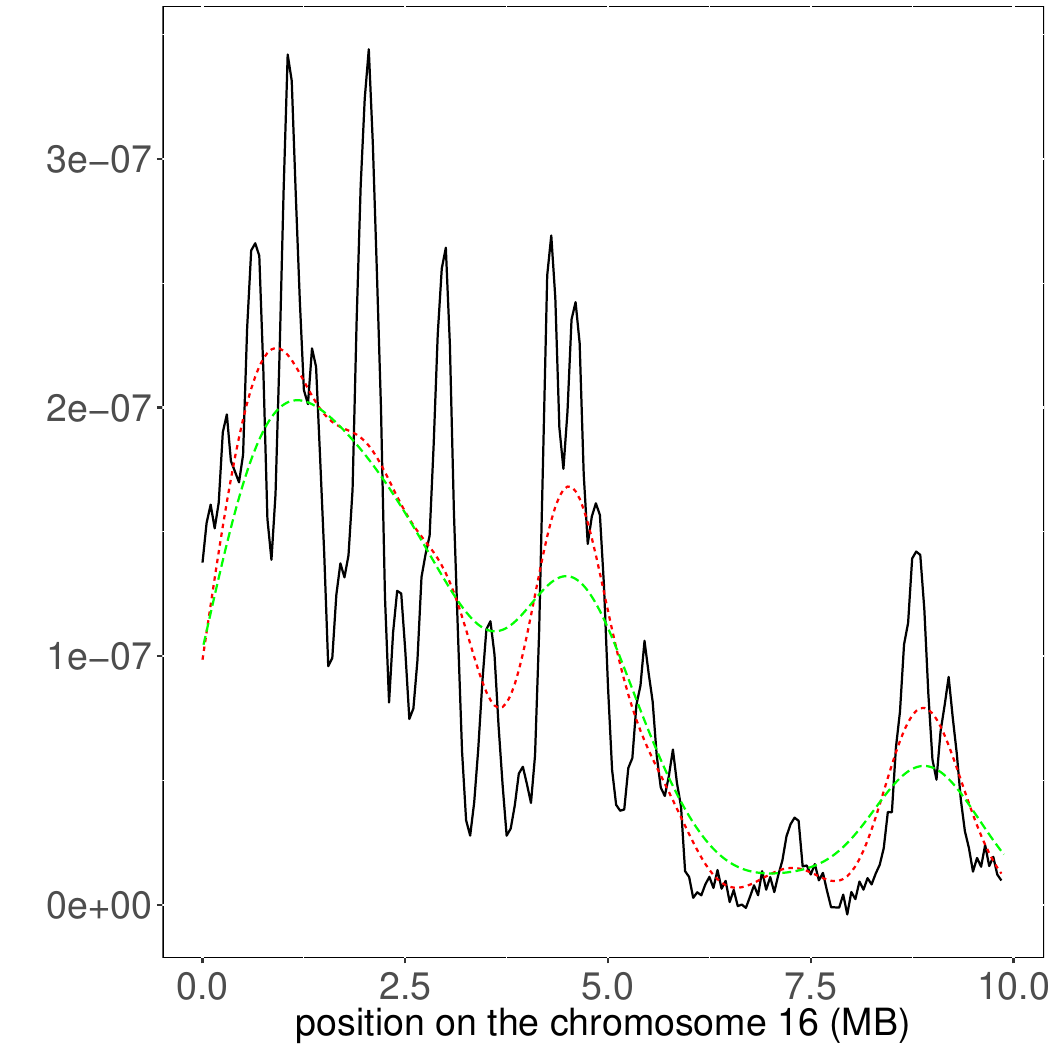}
	\end{center}
	\caption{ Estimation of the intensity of human replication origins along chromosome 16 ($N^+=874$), for 3 procedures:  our procedure (Epanechnikov kernel calibrated by our data-driven procedure, black plain line), the two procedures implemented in the $\texttt{fDKDE}$ package (red dashed line) and the procedure implemented in the $\texttt{R}$ $\texttt{density}$ function (Epanechnikov kernel) with the bandwidth calibrated by cross-validation (green dashed line).}
\end{figure}

In a second step we focused on the spatial distribution of G-quadruplex motifs that were shown to be associated with replication initiation in vertebrates \citep{PCA14}. Their precise role in replication remains unknown, and their effect may be associated with some epigenetic response \citep{HBL16} which makes their positional information very valuable regarding the biophysics constraints characterizing the DNA molecule. Thus we considered the spatial distribution of G-quadruplexes \citep{ZZH2020} along all replication origins \citep{PCA14}, by considering the initiation peak as the reference position (Figure \ref{fig:appli_g4}). When estimating the spatial distribution of G-quadruplex motifs around replication origins,
the estimator without deconvolution provides an almost flat estimated density with one small central peak. The $\texttt{fDKDE}$ estimator highlights one central peak and several smaller peaks. Finally, our estimator reveals a periodic clustering pattern of G-quadruplexes occurrences along replication origins, which is completely masked when computing a standard density estimator and only slightly suggested with the $\texttt{fDKDE}$ estimator. This clustering pattern could be related to the periodic organization of nucleosomes and chromatin around replication origins as suggested by experimental evidence \cite{pmid31332171}. Hence, our estimator provides a finer-scale resolution for the accumulation pattern of G-quadruplexes in the vicinity of replication initiation sites that could be biologically relevant.

\begin{figure}
	\begin{center}
		\includegraphics[scale=0.5,trim=0mm 0mm 0mm 0mm]{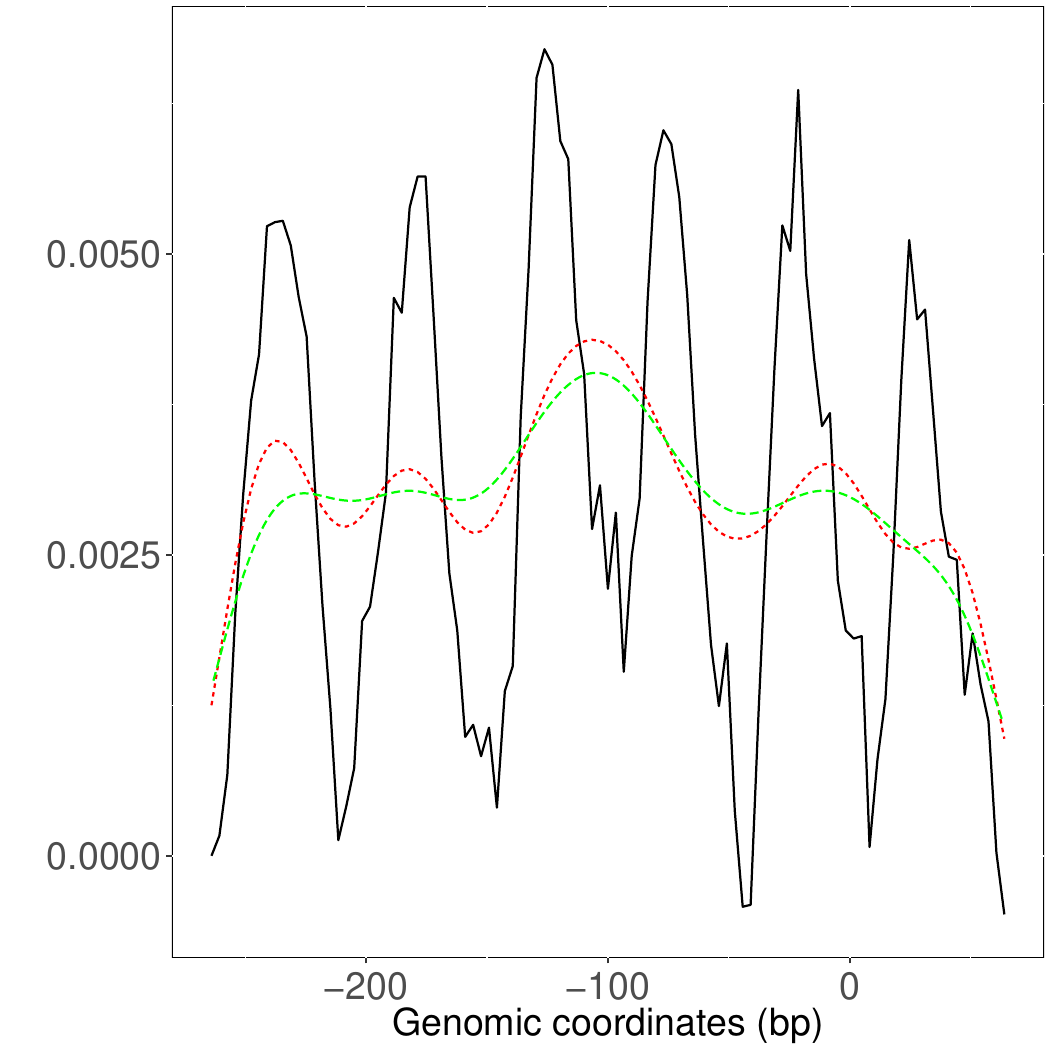}
	\end{center}
	\caption{Estimation of the intensity of G-quadruplexes around human replication origins along chromosome 1 ($a=25.9$ the average length of motifs, $N_+ = 2559$) for 3 procedures:  our procedure (Epanechnikov kernel calibrated by our data-driven procedure, black plain line), the two procedures implemented in the $\texttt{fDKDE}$ package (red dashed line) and the procedure implemented in the $\texttt{R}$ $\texttt{density}$ function (Epanechnikov kernel) with the bandwidth calibrated by cross-validation (green dashed line). The reference position 0 indicates the position of the replication initiation peak, and we compute the spatial distribution of Gquadruplexes upstream or downstream (negative position values) of origins.}
	\label{fig:appli_g4}
\end{figure}

\newpage
\section{Proofs}\label{sec:proofs}
If $\theta$ is a vector of constants (for instance $\theta=(T,a,K))$, we denote by $\square_\theta$ a positive constant that only depends on $\theta$ and that may change from line to line.

\medskip

\noindent In the sequel, we use at several places the following property: Setting
$$S_k:\ x\longmapsto \sum_{i=1}^{N_+}K_h'\bigg( x-(2k+1)a-Y_i \bigg),$$
since $Ah\leq a$, $S_k$ and $S_{k'}$ have disjoint supports if $k\not=k'$.


\subsection{Proof of Lemma~\ref{lem:varsym}}
\begin{proof}
	Considering first $\widehat f_h$, we have:
	\begin{align*}
		\E[\widehat f_h(x)]&
		=\frac{2a}{nh^2}\sum_{k=0}^{+\infty}\int_{\R}  K'\left(\frac{x-(2k+1)a-u}{h}\right)nf_Y(u)du\\
		&=\frac{1}{h^2}\sum_{k=0}^{+\infty}\int_{\R}  K'\left(\frac{x-(2k+1)a-u}{h}\right)[F_X(u+a)-F_X(u-a)]du\\
		&=\frac{1}{h^2}\sum_{k=0}^{+\infty}\int_{\R}  K'\left(\frac{x-2ka-v}{h}\right)F_X(v)dv-\int_{\R}  K'\left(\frac{x-2(k+1)a-v}{h}\right)F_X(v)dv
		\\
		&=\frac{1}{h^2}\int_{\R}  K'\left(\frac{x-v}{h}\right)F_X(v)dv\\
		&=\frac{1}{h}\int_{\R}  K\left(\frac{x-v}{h}\right)f_X(v)dv=(K_h\star f_X)(x).
	\end{align*}
	The first point is then straightforward by using the definition of $\widetilde f_h$. For the second point, observe that 
	\begin{align*}
		\widetilde f_h(t)-\E[\widetilde f_h(t)]&=
		\frac{a}{nh^2}\sum_{k=-\infty}^{+\infty}s_k\int K'\left(\frac{t-(2k+1)a-u}{h}\right)\left[dN^Y_u-nf_Y(u)du\right]\\
		&=\frac{a}{nh}\int_\R L_h(t-u)\left[dN^Y_u-nf_Y(u)du\right],
	\end{align*}
	with
	$$L_h(x):=\frac{1}{h}\sum_{k=-\infty}^{+\infty} s_kK'\left(\frac{x-(2k+1)a}{h}\right).$$
	Using the support $[-A,A]$ of $K$, for each $x$
	\begin{align*}
		(L_h(x))^2&=\frac{1}{h^2}\left(\sum_{k=-\infty}^{+\infty}s_k K'\left(\frac{x-(2k+1)a}{h}\right)\right)^2\\
		&=\frac{1}{h^2}\sum_{k=-\infty}^{+\infty} K'^2\left(\frac{x-(2k+1)a}{h}\right)
	\end{align*}
	as soon as $Ah\leq a$. 
	We have:
	\begin{align*}
		\int_0^T\E[(\widetilde f_h(t)-\E[\widetilde f_h(t)])^2]dt
		&=\frac{a^2}{n^2h^2}\int_0^T\Var\left(\int_\R L_h(t-u)dN^Y_u\right)dt\\
		&=\frac{a^2}{n^2h^2}\int_0^T\int_\R L_h^2(t-u)nf_Y(u)dudt\\
		&=\frac{a^2}{nh^4}\int_0^T\int_\R\sum_{k=-\infty}^{+\infty}\left(K'\left(\frac{t-u-(2k+1)a}{h}\right)\right)^2f_Y(u)dudt\\
		&=\frac{a^2}{nh^4}\int_0^T\int_\R\left(K'\left(\frac{t-v}{h}\right)\right)^2\sum_{k=-\infty}^{+\infty}f_Y(v-(2k+1)a)dvdt\\
		&=\frac{a}{2nh^4}\int_0^T\int_\R\left(K'\left(\frac{t-v}{h}\right)\right)^2(\lim_{+\infty} F_X -\lim_{-\infty} F_X)dvdt,
	\end{align*}
	which yields
	$$
	\E\Big[\|\widetilde f_h(t)-\E[\widetilde f_h(t)]\|_{2,T}^2\Big]= \frac{aT\|f_X\|_1\|K'\|_2^2}{2nh^3}.
	$$
\end{proof}
\subsection{Auxiliary lemma}
Our procedure needs the following result.
\begin{lemma}\label{lem:inversion}
	For any $h,t\in\mathcal H$,
	$$\widetilde f_{h,t}=\widetilde f_{t,h}.$$
\end{lemma}
\begin{proof}
	Since $K_h'(x)=(1/h^2)K'(x/h)$, we can write
	\begin{align*} 
		K_h\star \widehat f_t &
		=K_h\star\left(\frac{2a}{n}\sum_{k=0}^{+\infty}\sum_{i=1}^{N_+}(K_t)'({x-(2k+1)a-Y_i})\right)\\
		&=\frac{2a}{n}\sum_{k=0}^{+\infty}\sum_{i=1}^{N_+}K_h\star(K_t)'\left({x-(2k+1)a-Y_i}\right)
	\end{align*}
	Using that $K_h\star(K_t)'=(K_h\star K_t)'=(K_h)'\star K_t$, we obtain 
	$K_h\star \widehat f_t=K_t\star \widehat f_h$.
	
	In the same way, we can prove 
	$K_h\star \check f_t=K_t\star \check f_h$ and then 
	$K_h\star \widetilde f_t=K_t\star \widetilde f_h$.
\end{proof}
\subsection{Proof of Theorem~\ref{theo:GL}}
\label{sec:preuvebornesup}
Remember that
$$\hat h:=\argmin_{h\in \mathcal H}\left\{{\mathcal A}(h)+\frac{c\sqrt{N_+}}{nh^{3/2}}\right\},$$
with 
$${\mathcal A}(h):=\max_{t\in \mathcal H}\left\{\|\widetilde f_t-\widetilde f_{h,t}\|_{2,T}-\frac{c\sqrt{N_+}}{nt^{3/2}}\right\}_+. $$
For any $h\in\mathcal H$,
$$\|\widetilde f_{\hat h}-f_X\|_{2,T}\leq A_1+A_2+A_3,$$
with 
$$A_1:=\|\widetilde f_{\hat h}-\widetilde f_{\hat h,h}\|_{2,T}\leq {\mathcal A}(h)+\frac{c\sqrt{N_+}}{n\hat h^{3/2}},$$
$$A_2:=\|\widetilde f_{ h}-\widetilde f_{\hat h,h}\|_{2,T}\leq {\mathcal A}(\hat h)+\frac{c\sqrt{N_+}}{n h^{3/2}},$$
and
$$A_3:=\|\widetilde f_{ h}-f_X\|_{2,T}.$$
By definition of $\hat h$, we have:
$$A_1+A_2\leq 2{\mathcal A}(h)+\frac{2c\sqrt{N_+}}{n h^{3/2}}.$$
Therefore, by setting
$$\zeta_n(h):=\sup_{t\in\mathcal H}\left\{\|(\widetilde f_{t,h}-\E[\widetilde f_{t,h}])-(\widetilde f_t-\E[\widetilde f_t])\|_{2,T}-\frac{c\sqrt{N_+}}{n t^{3/2}}\right\}_+,$$
we have:
\begin{align*}
	A_1+A_2&\leq 2\zeta_n(h)+2\sup_{t\in\mathcal H}\|\E[\widetilde f_{t,h}]-\E[\widetilde f_t]\|_{2,T}+\frac{2c\sqrt{N_+}}{n h^{3/2}}\\
	&\leq 2\zeta_n(h)+2\sup_{t\in\mathcal H}\|K_h\star K_t\star f_X-K_t\star f_X\|_{2,T}+\frac{2c\sqrt{N_+}}{n h^{3/2}}\\
	&\leq 2\zeta_n(h)+2\|K\|_1\|K_h\star f_X- f_X\|_{2,T}+\frac{2c\sqrt{N_+}}{n h^{3/2}}.
\end{align*}
Finally, since $(\alpha+\beta+\gamma)^2\leq 3\alpha^2+3\beta^2+3\gamma^2$,
\begin{align*}
	\E[(A_1+A_2)^2]&\leq 12\E[\zeta_n^2(h)]+12\|K\|_1^2\|K_h\star f_X- f_X\|_{2,T}^2+\frac{12c^2\E[N_+]}{n^2 h^3}\\
	&\leq 12\E[\zeta_n^2(h)]+12\|K\|_1^2\|K_h\star f_X- f_X\|_{2,T}^2+\frac{12c^2\|f_X\|_1}{n h^3}.
\end{align*}
For the last term, we obtain:
\begin{align*}
	\E[A_3^2]&=\E[\|\widetilde f_{ h}-f_X\|_{2,T}^2]\\
	&=\E\left[\|\widetilde f_h-\E[\widetilde f_h]\|_{2,T}^2\right]+\|K_h\star f_X-f_X\|_{2,T}^2\\
	&=\frac{aT\|f_X\|_1\|K'\|_2^2}{2nh^3}+\|K_h\star f_X-f_X\|_{2,T}^2.
\end{align*}
Finally, replacing $c$ with its definition, namely
$$c=(1+\eta)(1+\|K\|_1)\|K'\|_2\sqrt{\frac{aT}{2}},$$
we obtain: for any $h\in\mathcal H$,
\begin{align}\nonumber
	\E[\|\widetilde f_{\hat h}-f_X\|_{2,T}^2]&\leq 2\E[(A_1+A_2)^2]+2\E[A_3^2]\\
	\nonumber
	&\leq 2(1+12\|K\|_1^2)\|K_h\star f_X-f_X\|_{2,T}^2+
	C_1\frac{aT\|f_X\|_1\|K'\|_2^2}{2nh^3}	
	+24\E[\zeta_n^2(h)]\\
	&\leq 
C_1
	\E[\|\widetilde f_h-f_X\|_{2,T}^2]+24\E[\zeta_n^2(h)] \label{1},
\end{align}
by using Lemma~\ref{lem:varsym} and by denoting $C_1=2+24(1+\eta)^2(1+\|K\|_1)^2$. It remains to prove that $\E[\zeta_n^2(h)]$ is bounded by $\frac{1}{n}$ up to a constant. We have:
\begin{align*}
	\zeta_n(h)&\leq\sup_{t\in\mathcal H}\left\{\|\widetilde f_{t,h}-\E[\widetilde f_{t,h}]\|_{2,T}+\|\widetilde f_t-\E[\widetilde f_t]\|_{2,T}-\frac{c\sqrt{N_+}}{nt^{3/2}}\right\}_+\\
	&\leq\sup_{t\in\mathcal H}\left\{(\|K\|_1+1)\|\widetilde f_t-\E[\widetilde f_t]\|_{2,T}-\frac{c\sqrt{N_+}}{n t^{3/2}}\right\}_+\\
	&\leq (\|K\|_1+1) S_n,
\end{align*}
with
\begin{align*}
	S_n:=\sup_{t\in\mathcal H}\left\{\|\widetilde f_t-\E[\widetilde f_t]\|_{2,T}-\frac{(1+\eta)\|K'\|_2\sqrt{aTN_+}}{\sqrt{2}n t^{3/2}}\right\}_+.
\end{align*}
For $\alpha\in (0,1)$ chosen later, we compute:
$$A_n:=\E[S_n^2\1_{\{N_+\leq (1-\alpha)^2n\|f_X\|_1\}}].$$
Recall that (see the proof of Lemma~\ref{lem:varsym})
$$\widetilde f_t(x)=\frac{a}{nt} \int_{\R} L_t(x-u)dN_u^Y,\quad 
\text{ with } L_t(x)=\frac{1}{t}\sum_{k=-\infty}^{\infty}s_k K'\left(\frac{x-(2k+1)a}{t}\right).$$
Since $At\leq a$,
\begin{align*}
	\left(\int_{\R} \sum_{k=-\infty}^{+\infty}s_k K'\left(\frac{x-(2k+1)a-u}{t}\right)dN^Y_u\right)^2&=\left(\sum_{i=1}^{N+}\sum_{k=-\infty}^{\infty}s_k K'\left(\frac{x-(2k+1)a-Y_i}{t}\right)\right)^2\\
	&\leq N_+\sum_{i=1}^{N+}\sum_{k=-\infty}^{\infty}\left(K'\left(\frac{x-(2k+1)a-Y_i}{h}\right)\right)^2\\
	&\leq N_+^2\|K'\|_\infty^2,
\end{align*}
which yields 
\begin{align*}
	S_n^2&\leq 2\sup_{t\in\mathcal H}\|\widetilde f_t\|^2_{2,T}+ 2\sup_{t\in\mathcal H}\|\E[\widetilde f_t]\|^2_{2,T}\\
	&\leq\sup_{t\in\mathcal H} \frac{2a^2}{n^2t^4}\int_0^T\left(\int_{\R} \sum_{k=-\infty}^{+\infty}s_k K'\left(\frac{x-(2k+1)a-u}{t}\right)dN^Y_u\right)^2dx+2\sup_{t\in\mathcal H}\|K_t\star f_X\|^2_{2,T}\\
	&\leq \sup_{t\in\mathcal H}\frac{2a^2\|K'\|_\infty^2 TN_+^2}{n^2t^4}+2\sup_{t\in\mathcal H}\frac{\|K\|^2_2\|f_X\|_1^2}{t}.
\end{align*}

Therefore, since $t\in \H \Rightarrow t^{-1}\leq \delta n^{1/3}$,
$$A_n\leq \Big(2\delta^4a^2\|K'\|_\infty^2 Tn^{4/3}+2\delta \|K\|^2_2n^{1/3}\Big)\|f_X\|_1^2\times\P(N_+\leq (1-\alpha)^2n\|f_X\|_1).$$
and since $n \geq 1$ 
$$A_n\leq(2\delta^4a^2\|K'\|_\infty^2T+2\delta\|K\|_2^2)n^2\|f_X\|_1^2\P(N_+\leq (1-\alpha)^2n\|f_X\|_1).$$

To bound the last term, we use, for instance, Inequality (5.2) of \cite{RB} (with $\xi=(2\alpha-\alpha^2)n\|f_X\|_1$ and with the function $f\equiv -1$), which shows that there exists $\alpha'>0$ only depending on $\alpha$ such that
$$\P(N_+\leq (1-\alpha)^2n\|f_X\|_1)\leq \exp(-\alpha' n \|f_X\|_1).$$
This shows that there exists a positive constant $C_\alpha$ such that
$$A_n\leq(2\delta^4a^2\|K'\|_\infty^2T+2\delta\|K\|_2^2)\frac{C_{\alpha}}{n\|f_X\|_1}.$$
We now deal with
$$B_n:=\E[S_n^2\1_{\{N_+> (1-\alpha)^2n\|f_X\|_1\}}].$$
We take $\alpha=\min(\eta/2,1/4)$. This implies
$$(1+\eta)(1-\alpha)\geq 1+\frac{\eta}{4}$$
and
\begin{align*}
	B_n&=\E\left[\sup_{t\in\mathcal H}\left\{\|\widetilde f_t-\E[\widetilde f_t]\|_{2,T}-\frac{(1+\eta)\|K'\|_2\sqrt{aTN_+}}{\sqrt{2}n t^{3/2}}\right\}_+^2\1_{\{N_+> (1-\alpha)^2n\|f_X\|_1\}}\right]\\
	&\leq \E\left[\sup_{t\in\mathcal H}\left\{\|\widetilde f_t-\E[\widetilde f_t]\|_{2,T}-\frac{(1+\eta/4)\sqrt{aT}\|K'\|_2\sqrt{\|f_X\|_1}}{\sqrt{2n} t^{3/2}}\right\}_+^2\right]\\
	&\leq\int_0^{+\infty}\P\left(\sup_{t\in\mathcal H}\left\{\|\widetilde f_t-\E[\widetilde f_t]\|_{2,T}-\frac{(1+\eta/4)\sqrt{aT}\|K'\|_2\sqrt{\|f_X\|_1}}{\sqrt{2n} t^{3/2}}\right\}_+^2\geq x\right)dx\\
	&\leq\sum_{t\in\mathcal H}\int_0^{+\infty}\P\left(\left\{\|\widetilde f_t-\E[\widetilde f_t]\|_{2,T}-\frac{(1+\eta/4)\sqrt{aT}\|K'\|_2\sqrt{\|f_X\|_1}}{\sqrt{2n} t^{3/2}}\right\}_+^2\geq x\right)dx.
\end{align*}
To conclude, it remains to control for any $x>0$ the probability inside the integral. For this purpose, we use the following lemma.
\begin{lemma}\label{concentration}
	Let $\e>0$ and $h\in\mathcal H$ be fixed. For any $x>0$, with probability larger than $1-\exp(-x)$,
	\begin{small}
		$$\|\widetilde f_h-\E[\widetilde f_h]\|_{2,T}\leq(1+\e)\|K'\|_2\sqrt{\frac{aT\|f_X\|_1}{2nh^3}}+\sqrt{12x}\|K'\|_1\sqrt{\frac{\|f_X\|_1(T+4a)}{4nh^2}}+(1.25+32\e^{-1})x\frac{\|K'\|_2}{nh^{3/2}}\sqrt{\frac{a(T+4a)}{2}}.$$
	\end{small}
\end{lemma}
\begin{proof}
	We set:
	\begin{align*}
		U(t)&=\widetilde f_h(t)-\E[\widetilde f_h(t)]\\
		&=\frac{a}{nh^2}\sum_{k=-\infty}^{+\infty}s_k\int_\R K'\left(\frac{t-(2k+1)a-u}{h}\right)\left[dN^Y_u-nf_Y(u)du\right]\\
		&=\frac{a}{nh}\int_\R L_h(t-u)\left[dN^Y_u-nf_Y(u)du\right],
	\end{align*}
	with
	$$L_h(x):=\frac{1}{h}\sum_{k=-\infty}^{+\infty} s_kK'\left(\frac{x-(2k+1)a}{h}\right).$$
	Let $\mathcal D$ a countable dense subset of the unit ball of ${\mathbb L}_2[0,T]$. We have:
	\begin{align*}
		\|U\|_{2,T}&=\sup_{g\in\mathcal D}\int_0^Tg(t)U(t)dt\\
		&=\sup_{g\in\mathcal D}\int_\R\Psi_g(u)\left[dN^Y_u-nf_Y(u)du\right],
	\end{align*}
	with 
	$$\Psi_g(u):=\frac{a}{nh}\int_0^TL_h(t-u)g(t)dt=\frac{a}{nh}(\widetilde L_h\star g)(u)$$
	and $\widetilde L_h(x)=L_h(-x)$, where the convolution product is computed on $[0,T]$.
	We use Corollary~2 of \cite{RB}. So, we need to bound $\E[\|U\|_{2,T}]$ and
	$$v_0:=\sup_{g\in\mathcal D}\int_\R\Psi_g^2(u)nf_Y(u)du.$$
	We also have to determine $b$, a deterministic upper bound for all the $\Psi_g$'s. We have already proved in the proof of Lemma~\ref{lem:varsym} that 
	\begin{align*}
		\E[\|U\|_{2,T}^2]&= \frac{aT\|f_X\|_1\|K'\|_2^2}{2nh^3},
	\end{align*}
	which implies
	\begin{equation}\label{boundU}
		\E[\|U\|_{2,T}]\leq \|K'\|_2\sqrt{\frac{aT\|f_X\|_1}{2nh^3}}.
	\end{equation}
	If we denote $I(h,u):=\{k\in\mathbb{Z}: -u-Ah-a\leq2ka\leq Ah+T-u-a\},$
	then
	\begin{align*}
		\Psi_g^2(u)&=\frac{a^2}{n^2h^2}\left(\int_0^TL_h(t-u)g(t)dt\right)^2\\
		&\leq \frac{a^2}{n^2h^2}\int_0^TL_h^2(t-u)dt\times \int_0^Tg^2(t)dt\\
		&\leq\frac{a^2}{n^2h^4}\int_0^T\sum_{k\in I(h,u)}\left(K'\left(\frac{t-u-(2k+1)a}{h}\right)\right)^2dt\\
		&\leq \frac{a^2}{n^2h^3}\|K'\|_2^2\times\mbox{card}(I(h,u))\\
		&\leq \frac{a^2}{n^2h^3}\|K'\|_2^2\times(T/(2a)+Ah/a+1)
	\end{align*} 
	and we can set, under the condition on $\H$, 
	$$b:=\frac{a}{nh^{3/2}}\|K'\|_2\sqrt{\frac{T+4a}{2a}},$$
	which is negligible with respect to the upper bound of $\E[\|U\|_{2,T}]$ given in \eqref{boundU}. We now deal with
	$$v_0:=n\times\sup_{g\in\mathcal D}\int_\R\Psi_g^2(u)f_Y(u)du.$$ We have:
	\begin{align*}
		v_0&= \frac{a^2}{nh^2}\sup_{g\in\mathcal D}\int_\R\left(\int_0^TL_h(t-u)g(t)dt\right)^2f_Y(u)du\\
		&\leq\frac{a^2}{nh^2}\sup_{g\in\mathcal D}\int_\R\left(\int_0^T|L_h(t-u)|dt\int_0^T|L_h(t-u)|g^2(t)dt\right)f_Y(u)du.
	\end{align*}
	Since
	\begin{align*}
		\int_0^T|L_h(t-u)|dt&\leq\frac{1}{h}\int_0^T\sum_{k\in I(h,u)}\left|K'\left(\frac{t-u-(2k+1)a}{h}\right)\right|dt\\
		&\leq \|K'\|_1\mbox{card}(I(h,u))\\
		&\leq \|K'\|_1\frac{(T+4a)}{2a},
	\end{align*}
	we obtain
	\begin{align*}
		v_0&\leq\frac{a^2}{nh^2}\|K'\|_1\frac{(T+4a)}{2a}\sup_{g\in\mathcal D}\int_\R\int_0^T\sum_{k=-\infty}^{+\infty}\frac{1}{h}\left|K'\left(\frac{t-u-(2k+1)a}{h}\right)\right|g^2(t)dtf_Y(u)du\\
		&\leq\frac{a}{2nh^2}\|K'\|_1(T+4a)\sup_{g\in\mathcal D}\int_0^T\left(\int\frac{1}{h}\left|K'\left(\frac{t-v}{h}\right)\right|\sum_{k=-\infty}^{+\infty}f_Y(v-(2k+1)a)dv\right)g^2(t)dt\\
		&\leq\frac{a}{2nh^2}\|K'\|_1(T+4a)\sup_{g\in\mathcal D}\int_0^T\left(\int\frac{1}{h}\left|K'\left(\frac{t-v}{h}\right)\right|dv\right)\frac{\|f_X\|_1}{2a}g^2(t)dt\\
		&\leq\frac{\|f_X\|_1}{4nh^2}\|K'\|_1^2(T+4a).
	\end{align*}
	Inequality (5.7) of \cite{RB} yields, for any $x>0$,
	$$
	\P\left(\|U\|_{2,T}\geq (1+\e)\E[\|U\|_{2,T}]+\sqrt{12v_0x}+(1.25+32\e^{-1})bx\right)\leq \exp(-x).
	$$
Setting
	\begin{small}
		\begin{align*}
			RHS&:=(1+\e)\E[\|U\|_{2,T}]+\sqrt{12v_0x}+(1.25+32\e^{-1})bx,
			\end{align*}
			we obtain
		\begin{align*}			
			RHS&\leq (1+\e)\|K'\|_2\sqrt{\frac{aT\|f_X\|_1}{2nh^3}}+\sqrt{12x}\|K'\|_1\sqrt{\frac{\|f_X\|_1(T+4a)}{4nh^2}}+(1.25+32\e^{-1})x\frac{\|K'\|_2}{nh^{3/2}}\sqrt{\frac{a(T+4a)}{2}}.
		\end{align*}
	\end{small}
\end{proof}
The previous lemma states that for any sequence of weights $(w_h)_{h\in{\mathcal H}}$, setting $x=w_h+u$, with $u>0$, with probability larger than $1-\exp(-u)\sum_{h\in{\mathcal H}}\exp(-w_h)$, for all $h\in\mathcal H$,
$$
\|\widetilde f_h-\E[\widetilde f_h]\|_{2,T}\leq M_h+\sqrt{12u}\|K'\|_1\sqrt{\frac{\|f_X\|_1(T+4a)}{4nh^2}}+(1.25+32\e^{-1})u\frac{\|K'\|_2}{nh^{3/2}}\sqrt{\frac{a(T+4a)}{2}}$$
with
\begin{align*}
	M_h&:=(1+\e)\|K'\|_2\sqrt{\frac{aT\|f_X\|_1}{2nh^3}}+\sqrt{12w_h}\|K'\|_1\sqrt{\frac{\|f_X\|_1(T+4a)}{4nh^2}}+(1.25+32\e^{-1})w_h\frac{\|K'\|_2}{nh^{3/2}}\sqrt{\frac{a(T+4a)}{2}}\\
	&=\|K'\|_2\sqrt{\frac{aT\|f_X\|_1}{2nh^3}}\left(1+\e+\sqrt{12w_hh}\frac{\|K'\|_1}{\|K'\|_2}\frac{\sqrt{T+4a}}{\sqrt{2aT}}+\frac{(1.25+32\e^{-1})w_h}{\sqrt{n\|f_X\|_1}}\sqrt{\frac{T+4a}{T}}\right)\\
	&\leq \frac{(1+\eta/4)\sqrt{aT}\|K'\|_2\sqrt{\|f_X\|_1}}{\sqrt{2n} h^{3/2}},
\end{align*}
for $\e=\eta/8$
 and for $n$ large enough, by taking $w_h=h^{-1/2}|\log h|^{-1}$ for instance, since in this case,
$$w_h h=o(1)\quad\mbox{and}\quad h^{-1}=O(n\|f_X\|_1).$$
Therefore, 
\begin{align*}
	B_n&\leq\sum_{h\in\mathcal H}\int_0^{+\infty}\P\left(\left\{\|\widetilde f_h-\E[\widetilde f_h]\|_{2,T}-M_h\right\}_+^2\geq x\right)dx.
\end{align*}
By setting $u$ such that
$$x:=(g(u))^2=\left(\sqrt{12u}\|K'\|_1\sqrt{\frac{\|f_X\|_1(T+4a)}{4nh^2}}+(1.25+32\e^{-1})u\frac{\|K'\|_2}{nh^{3/2}}\sqrt{\frac{a(T+4a)}{2}}\right)^2,$$
so
$$dx=2g(u)\times\left(\frac{\sqrt{12}}{2\sqrt{u}}\|K'\|_1\sqrt{\frac{\|f_X\|_1(T+4a)}{4nh^2}}+(1.25+32\e^{-1})\frac{\|K'\|_2}{nh^{3/2}}\sqrt{\frac{a(T+4a)}{2}}\right)du$$
and using that $\int_0^{\infty} e^{-u}(D \sqrt{u}+E u)^2u^{-1}du\leq 2D^2+2E^2$, we obtain
\begin{align*}
	B_n&\leq\sum_{h\in\mathcal H}\int_0^{+\infty}e^{-(w_h+u)}\times 2g(u)\times\left(\frac{\sqrt{12}}{2\sqrt{u}}\|K'\|_1\sqrt{\frac{\|f_X\|_1(T+4a)}{4nh^2}}+(1.25+32\e^{-1})\frac{\|K'\|_2}{nh^{3/2}}\sqrt{\frac{a(T+4a)}{2}}\right)du\\
	&\leq 2\sum_{h\in\mathcal H}e^{-w_h}\int_0^{+\infty}e^{-u}(g(u))^2u^{-1}du\\	
&\leq	4\sum_{h\in\mathcal H}e^{-w_h}\left(
12\|K'\|_1^2\frac{\|f_X\|_1(T+4a)}{4nh^2}+
(1.25+32\e^{-1})^2\frac{\|K'\|_2^2}{n^2h^{3}}{\frac{a(T+4a)}{2}}
\right)	.
\end{align*}
Since $\sum_{h\in\mathcal H}e^{-w_h}h^{-2}$ and $\sum_{h\in\mathcal H}e^{-w_h}h^{-3}$ are bounded by an absolute constant, say $C$, we can write, still for $\e=\eta/8$,
\begin{align*}
	B_n
&\leq	4C\left(
12\|K'\|_1^2\frac{\|f_X\|_1(T+4a)}{4n}+
(1.25+32(\eta/8)^{-1})^2\frac{\|K'\|_2^2}{n^2}{\frac{a(T+4a)}{2}}
\right)	\\
	&\leq 12C\|K'\|_1^2(T+4a)\frac{\|f_X\|_1}{n}+2C(1.25+256/\eta)^2\|K'\|_2^2a(T+4a)\frac{1}{n^2}.
\end{align*}
Finally, we obtain
$$\E[\zeta_n^2(h)]\leq (\|K\|_1+1)^2\E[S_n^2]
\leq  (\|K\|_1+1)^2\left(\frac{c_1}{n\|f_X\|_1}+ \frac{c_2\|f_X\|_1}{n}+\frac{c_3}{n^2}\right)
,$$
with $c_1=C_{\alpha}(2\delta^4a^2\|K'\|_\infty^2T+2\delta\|K\|_2^2)$,
$c_2=12C\|K'\|_1^2(T+4a)$ and $c_3=2C(1.25+256/\eta)^2\|K'\|_2^2a(T+4a).$
This concludes the proof of the theorem, with 
\begin{equation}\label{C2}
C_2=(\|K\|_1+1)^2\left(c_1\|f_X\|_1^{-1}+ {c_2\|f_X\|_1}+c_3\right).
\end{equation}
\color{black}

\subsection{Proof of Theorem~\ref{t2}}\label{prooflowerbound}
To prove Theorem~\ref{t2}, without loss of generality, we assume that $T$ is a positive integer. We denote $a\wedge b =\min (a,b)$ and $a\vee b =\max(a,b)$. The cardinal of a finite set $m$ is denoted by $|m|$.

As usual in the proofs of lower bounds, we build a set of intensities $(f_m)_{m\in \M}$ quite distant from each other in terms of the ${\mathbb L}_2$-norm, but whose distance between the resulting models is small. This set of intensities is based on wavelet expansions. More precisely, let $\psi$ be the Meyer wavelet built with with $C^2$-conjugate mirror
filters (see for instance Section~7.7.2 of \cite{Mallat}). We shall use  in particular that $\psi$ is $C^{\infty}$ and there exists a positive constant $c_\psi$ such that 
\begin{enumerate}
	\item $|\psi(x)|\leq c_\psi(1+|x|)^{-2}$ for any $x\in\R$,
	\item $\psi^*$ is $C^2$ and $\psi^*$ has support included into $[-8\pi/3,-2\pi/3]\cup[2\pi/3,8\pi/3]$,
\end{enumerate}
where $\psi^*(\xi)=\int e^{it\xi}\psi(x)dx$ is the Fourier transform of $\psi$. Observe that this implies that the functions $$\xi\mapsto \psi^*(\xi),\quad \xi\mapsto \psi^*(\xi)\xi^{-1},\quad \xi\mapsto \psi^*(\xi)\xi^{-2},\quad \xi\mapsto  (\psi^*)'(\xi) \quad\mbox{and}\quad \xi\mapsto  (\psi^*)'(\xi)\xi^{-1}$$ are  bounded by a constant. Without loss of generality, we assume that this constant is $c_\psi$.

Let 
$$f_1(x)=\frac{c_1}{1+x^2},$$
where $c_1$ is a positive constant small enough, which is chosen such that $f_1$ belongs to 
${\mathcal S}^\beta(L/2)$, where we denote
$${\mathcal S}^\beta(L):={\mathcal S}^\beta(L,0,+\infty)=\left\{g\in\L{2}:\quad \int_{-\infty}^{+\infty}|g^*(\xi)|^2(\xi^2+1)^\beta d\xi\leq L^2 \right\}.$$
Indeed, note that 
\begin{equation}\label{cbeta}
{\mathfrak c_{\beta}}^2:=\int |c_1^{-1}f_1^*(\xi)|^2(\xi^2+1)^{\beta}=\int  \pi^2\exp (-2|\xi|)(\xi^2+1)^{\beta}<\infty
\end{equation}
so that it is sufficient to choose $c_1=\mathfrak c_{\beta}^{-1} L/2$. 
With this choice we also have $r\leq\|f_1\|_1=c_1\pi\leq bL$ since we have assumed $rL^{-1}\leq \pi/(2\mathfrak c_\beta)\leq b$; then $f_1\in {\mathcal S}^\beta(L/2,r,b)$.

We recall a combinatorial lemma due to Birg\'e and Massart (see Lemma 8 in \cite{RB}, see also Lemma 2.9 in \cite{tsybakov}).
\begin{lemma}\label{Gallager}
	Let $D$ an integer and $\varGamma$ be a finite set with cardinal $D$. There exist  absolute constants $\tau$ and $\sigma$ such that there exists $\mathcal{M}_D\subset \mathcal{P}(\varGamma)$, satisfying $\log|\mathcal{M}_D|\geq \sigma D$ and such that
	for all distinct sets $m$ and $m'$ belonging to $\mathcal{M}_D$ the symmetric difference of  $m$ and $m'$, denoted $m\Delta m'$, satisfies  $|m\Delta m'|\geq \tau D$.
\end{lemma}
Here we choose $\varGamma:=\{0,\dots,D-1\}$ with $D:=T2^{j-1}$
where $j$ is an integer to be chosen later (so, we take $T2^{j-1}\geq 1$), and we denote $\M:=\mathcal{M}_D$ given in the previous lemma. Thus $\log|\M|\geq \sigma T2^{j-1}$ and for all $m,m'\in \M$ :  $\tau T2^{j-1}\leq |m\Delta m'|\leq T2^{j-1}$.

Now, for $a_j>0$ to be chosen, for $m\in \M$, for $x\in\R$, we set
$$f_m(x):=f_1(x)+a_j\sum_{k\in m}\psi_{jk}(x),$$
where we have denoted, as usual, $\psi_{jk}(x):=2^{j/2}\psi(2^jx-k)$.

We compute 
$\psi_{jk}^*(\xi)
=2^{-j/2}e^{i\xi k2^{-j}}\psi^*(\xi 2^{-j}),$
which gives
\begin{align*}
	\int |(f_m-f_1)^*(\xi)|^2(1+\xi^2)^{\beta}d\xi&=\int \left| a_j2^{-j/2}\psi^*(\xi2^{-j})\sum_{k\in m }  e^{i\xi k2^{-j}}\right|^2 (1+\xi^2)^{\beta}d\xi\\
	&=a_j^2\int \left| \psi^*(t)\sum_{k\in m }  e^{ikt}\right|^2 (1+t^22^{2j})^{\beta}dt\\
	&\leq \square_{\psi,\beta}a_j^22^{2j\beta}\int_{-3\pi}^{3\pi}\left|\sum_{k\in m }  e^{ikt}\right|^2 dt\\
	&\leq \square_{\psi,\beta}2^{2j\beta}|m|a_j^2\leq  \square_{\psi,\beta}T 2^{j(2\beta+1)}a_j^2,
\end{align*}
using Parseval's theorem and $|m|\leq D=T2^{j-1}$.
\color{black}
We assume from now on that 
\begin{equation}\label{cont1}
	T2^{j(2\beta+1)} a_j^2 \leq C(\psi,\beta)L^{2}
\end{equation}
for $C(\psi,\beta)$ a constant only depending on $\beta$ and $L$ small enough, so that $(f_m-f_1)$ belongs to ${\mathcal S}^\beta(L/2)$ and then 
$f_m\in {\mathcal S}^\beta(L)$.

Let us verify that $f_m$ is non-negative, and then is an intensity of a Poisson process. 
Since $\|f_m\|_1=\int f_m(x)dx=\|f_1\|_1 \in[r;bL]$, this will also ensure that $f_m\in {\mathcal S}^\beta(L,r,b)$.
For any real~$x$,
\begin{align*}
	\frac{f_m(x)-f_1(x)}{f_1(x)}&=c_1^{-1}(1+x^2) a_j2^{j/2} \sum_{k\in m}\psi(2^jx-k).
\end{align*}
Recall that $\psi(x)\leq c_\psi(1+|x|)^{-2}$. Let us now study 3 cases.
\begin{enumerate}
	\item If $0\leq |x|\leq T+1$, we have:
	$$(1+x^2) \left|\sum_{k\in m}\psi(2^jx-k)\right|\leq (T^2+2T+2)c_\psi\sum_{k\in m}(1+|2^jx-k|)^{-2}\leq 2c_\psi(T^2+2T+2)\sum_{\ell=1}^{+\infty}\ell^{-2},$$
	and the last upper bound is smaller than a finite constant only depending on $T$ and $c_\psi$.
	\item If $x\geq T+1$, since $|m|\leq D=T2^{j-1}$, we have:
	$$ (1+x^2)\left|\sum_{k\in m}\psi(2^jx-k)\right|\leq c_\psi T2^{j-1}(1+2^j(x-T))^{-2}(1+x^2)\leq c_\psi T2^{-j-1}\sup_{x\geq T+1}\frac{1+x^2}{(x-T)^2},$$
	and the last expression is smaller than a finite constant only depending on $T$ and $c_\psi$.
	\item If $x\leq -T-1$,
	$$ (1+x^2)\left|\sum_{k\in m}\psi(2^jx-k)\right|\leq c_\psi T2^{j-1}(1+2^j(-x))^{-2}(1+x^2)\leq c_\psi T2^{-j-1}\sup_{x\leq -T-1}\frac{1+x^2}{(-x)^2},$$
\end{enumerate}
Finally we obtain that there exists $\bar C(T,c_\psi)$ a constant only depending on $T$ and $c_\psi$ such that
\begin{align*}
	\frac{|f_m(x)-f_1(x)|}{f_1(x)}
	&\leq c_1^{-1}  a_j2^{j/2}\bar C(T,c_\psi).
\end{align*}
We take $a_j$ such that 
\begin{equation}\label{cont2}
	c_1^{-1}  a_j2^{j/2}\bar C(T,c_\psi)\leq \frac12.
\end{equation}
This ensures that $f_m\geq \frac12 f_1\geq 0.$
Another consequence is that $f_\varepsilon \star f_m\geq \frac12 f_\varepsilon \star f_1$. This provides
\begin{align*}
	f_\varepsilon \star f_m(x) &\geq \frac12\int_{-a}^{a}\frac{1}{2a} \frac{c_1}{1+(x-t)^2}dt \geq \frac12\:\frac{c_1}{1+(|x|+a)^2}\\
	&\geq \frac12\:\frac{c_1}{1+2a^2+2x^2} \geq \frac{c_2^{-1}}{ 1+x^2},
\end{align*}
denoting $c_2=c_2(a,\beta, L)=\max(4,2+4a^2)/c_1$.

Finally, we evaluate the distance between the distributions of the observations $N^Y$ when $N^X$ has intensity  $nf_m$ and $nf_{m'}$. We denote by $\P_m$ the probability measure associated with $N^Y$, which has  intensity 
$g_m:=f_\varepsilon\star nf_m$,
and we denote by $K(\P_m,\P_{m'})$ the Kullback-Leibler divergence between $\P_m$ and $ \P_{m'}$. Using \cite{cavalierkoo}, we have
$$K(\P_m,\P_{m'})=\int g_m(x) \phi\left(\log\left(\frac{g_{m'}(x)}{g_m(x)}\right)\right)dx $$
where for any $x\in\R$, $\phi(x)=\exp(x)-x-1.$ Since for any $x
>-1,$ $\log(1+x)\geq x/(1+x)$, we have
$$K(\P_m,\P_{m'})\leq \int \frac{(g_m(x)-g_{m'}(x))^2}{g_m(x)}dx=n \int \frac{((f_\varepsilon\star f_m)(x)-(f_\varepsilon\star f_{m'})(x))^2}{(f_\varepsilon\star f_m)(x)}dx.$$
For $m$ and $m'$ in $\M$,  denote 
$$\theta(x)=a_j^{-1}(f_\varepsilon\star (f_m-f_{m'}))(x)=
\sum_{k\in m \Delta m'}b_k (f_\varepsilon \star \psi_{jk})(x)$$
with $b_k=1$ if $k\in m$ and $b_k=-1$ if $k\in m'$. 
Denote also 
$\theta^*(\xi)=\int e^{i\xi x}\theta(x)dx$ its Fourier transform, and 
$(\theta^*)'(\xi)=\int ix e^{i\xi x}\theta(x)dx$ the derivative of $\theta^*$. 
Parseval's theorem gives
$$\|\theta\|_2^2=\frac{1}{2\pi}\|\theta^*\|_2^2,\quad\text{ and }
\|x\theta(x)\|_2^2=\frac{1}{2\pi}\|(\theta^*)'\|_2^2.$$
Thus
\begin{align*}
	\frac1nK(\P_m,\P_{m'})&\leq  \int \frac{((f_\varepsilon\star f_m)(x)-(f_\varepsilon\star f_{m'})(x))^2}{(f_\varepsilon\star f_m)(x)}dx\leq c_2\int (1+x^2)(f_\varepsilon\star (f_m-f_{m'})(x))^2dx\\
	&\leq  c_2a_j^2\int (1+x^2)\theta(x)^2dx
	\leq  \frac{c_2}{2\pi} a_j^2\left(\|\theta^*\|_2^2+\|(\theta^*)'\|_2^2\right).
\end{align*}
We recall that 
$\psi_{jk}^*(\xi)
=2^{-j/2}e^{i\xi k2^{-j}}\psi^*(\xi 2^{-j}),$
which gives
\begin{align*}
	\theta^*(\xi)&=\sum_{k\in m \Delta m'}b_kf_\varepsilon^*(\xi)\psi_{jk}^*(\xi)
	=\sum_{k\in m \Delta m'}b_kf_\varepsilon^*(\xi)2^{-j/2}e^{i\xi k2^{-j}}\psi^*(\xi 2^{-j})\\
	&=2^{-j/2}f_\varepsilon^*(\xi)\psi^*(\xi 2^{-j})\sum_{k\in m \Delta m'} b_ke^{i\xi k2^{-j}}.
\end{align*}
Thus, remembering that for $\xi\in\R,$
\begin{equation}\label{Fourierf}
	|f_\varepsilon^*(\xi)|=\left|\frac{\sin(a\xi)}{a\xi}\right|\leq \min (1,|a\xi|^{-1}),
\end{equation}
we have
\begin{align*}
	\|\theta^*\|_2^2&=\int \left| 2^{-j/2}f_\varepsilon^*(\xi)\psi^*(\xi2^{-j})\sum_{k\in m \Delta m'} b_k e^{i\xi k2^{-j}}\right|^2 d\xi\\
	&=\int \left| f_\varepsilon^*(u2^j)\psi^*(u)\sum_{k\in m \Delta m'} b_ke^{iku}\right|^2 du\\
	&\leq \int  \min (1,|a2^ju|^{-2}){|\psi^*(u)|^2}\left|\sum_{k\in m \Delta m'} b_ke^{iku}\right|^2 du\\
	&\leq \int_{-8\pi/3}^{8\pi/3}  \min (1,|a2^j|^{-2})c_{\psi}^2\left|\sum_{k\in m \Delta m'} b_ke^{iku}\right|^2 du
\end{align*}
using the properties of $\psi$. 
Parseval's theorem gives
$$\frac{1}{2\pi}\int_{-\pi}^{\pi} \left|\sum_{k\in m \Delta m'} b_k e^{iku}\right|^2 du=\sum_{k\in m \Delta m'}b_k^2
=|m\Delta m'|\leq T2^{j-1}.$$
Then 
$$\|\theta^*\|_2^2\leq  3\pi c_\psi^2T2^{j}(a^{-2}2^{-2j}\wedge 1).$$ 
Let us now bound $\|(\theta^*)'\|_2^2$.
First $$(\psi_{jk}^*)'(\xi )
=2^{-3j/2}e^{i\xi k2^{-j}}\left(\psi^*(\xi 2^{-j})ik+(\psi^*)'(\xi 2^{-j}) \right),
$$
then 
\begin{align*}
	(\theta^*)'(\xi )&=\sum_{k\in m \Delta m'}b_k\left[(f_\varepsilon^*)'(\xi )\psi_{jk}^*(\xi )+
	f_\varepsilon^*(\xi )(\psi_{jk}^*)'(\xi )\right]\\
	&=\sum_{k\in m \Delta m'}b_k(f_\varepsilon^*)'(\xi )2^{-j/2}e^{i\xi k2^{-j}}\psi^*(\xi 2^{-j})+b_kf_\varepsilon^*(\xi )2^{-3j/2}e^{i\xi k2^{-j}}\left(\psi^*(\xi 2^{-j})ik+(\psi^*)'(\xi 2^{-j}) \right)\\
	&=\alpha_1(\xi )+\alpha_2(\xi )+\alpha_3(\xi )
\end{align*}
where
\begin{align*}
	\alpha_1(\xi )=2^{-j/2}(f_\varepsilon^*)'(\xi )\psi^*(\xi 2^{-j})\sum_{k\in m \Delta m'}b_ke^{i\xi k2^{-j}},\\
	\alpha_2(\xi )=2^{-3j/2}f_\varepsilon^*(\xi )\psi^*(\xi 2^{-j})\sum_{k\in m \Delta m'}ikb_ke^{i\xi k2^{-j}},\\
	\alpha_3(\xi )=2^{-3j/2}f_\varepsilon^*(\xi )(\psi^*)'(\xi 2^{-j})\sum_{k\in m \Delta m'}b_ke^{i\xi k2^{-j}}.
\end{align*}
Reasoning as before, and using that 
\begin{equation}\label{majo-derivee}
	|(f_\varepsilon^*)'(\xi)|=\left|\frac{\cos(a\xi)}{\xi}-\frac{1}{\xi}\times\frac{\sin(a\xi)}{a\xi}\right|\leq \frac{2}{|\xi|},
\end{equation}
we can write
\begin{align*}
	\|\alpha_1\|_2^2
	&=\int \left| (f_\varepsilon^*)'(u2^j)\psi^*(u)\sum_{k\in m \Delta m'} b_ke^{iku}\right|^2 du\\
	&\leq \int\frac{4}{u^22^{2j}}|\psi^*(u)|^2 \left|\sum_{k\in m \Delta m'} b_ke^{iku}\right|^2 du\\
	&\leq 4c_\psi^2 2^{-2j} \times 6\pi |m\Delta m'|\\
	&\leq 12  \pi  c_\psi^2  T2^{-j}.
\end{align*}
In the same way, using \eqref{Fourierf},
\begin{align*}
	\|\alpha_2\|_2^2
	&=2^{-2j}\int \left| f_\varepsilon^*(u2^j)\psi^*(u)\sum_{k\in m \Delta m'} ikb_ke^{iku}\right|^2 du\\
	&\leq 2^{-2j}\int(1\wedge a^22^{-2j}u^{-2})|\psi^*(u)|^2\left|\sum_{k\in m \Delta m'} ikb_ke^{iku}\right|^2 du\\
	&\leq  3c_\psi^22^{-2j}(a^{-2}2^{-2j}\wedge 1)\int_{-\pi}^{\pi} \left|\sum_{k\in m \Delta m'} ikb_ke^{iku}\right|^2 du\\
	&\leq  6\pi c_\psi^22^{-2j}(a^{-2}2^{-2j}\wedge 1) \sum_{k\in m \Delta m'} k^2
\end{align*}
and we obtain that
$$\|\alpha_2\|_2^2\leq  cc_\psi^2T^32^{j}(a^{-2}2^{-2j}\wedge 1),$$
for $ c$ an absolute constant. Similarly,
\begin{align*}
	\|\alpha_3\|_2^2
	&=2^{-2j}\int \left| f_\varepsilon^*(u2^j)(\psi^*)'(u)\sum_{k\in m \Delta m'} b_ke^{iku}\right|^2 du\\
	&\leq 2^{-2j}\int (1\wedge a^22^{-2j}u^{-2})|(\psi^*)'(u)|^2\left|\sum_{k\in m \Delta m'} b_ke^{iku}\right|^2 du\\
	&\leq  3c_\psi^22^{-2j}(a^{-2}2^{-2j}\wedge 1)\int_{-\pi}^{\pi} \left|\sum_{k\in m \Delta m'} b_ke^{iku}\right|^2 du\\
	&\leq3\pi c_\psi^2T2^{-j}(a^{-2}2^{-2j}\wedge 1).
\end{align*}
Finally, since $a$ is smaller than an absolute constant and $T$ is larger than an absolute constant, we have that 
$$K(\P_m,\P_{m'}) \leq C{c_2}c_\psi^2na_j^2T^32^{j}(a^{-2}2^{-2j}\wedge 1),$$
for $C$ an absolute constant.

Now, let us give the following version of Fano's lemma, derived from \cite{birFano}.
\begin{lemma}\label{fano}
	Let $(\P_i)_{i\in\{0,\ldots,I\}}$ be a finite family of probability measures defined on the same measurable space $\Omega$. One sets
	$$\overline{K}_I=\frac{1}{I}\sum_{i=1}^I K(\P_i,\P_0).$$
	Then, there exists an absolute constant $B$ ($B=0.71$ works)  such that
	if $Z$ is a random variable on $\Omega$ with values in
	$\{0,\ldots,I\}$, one has
	$$\inf\limits_{0\leq i\leq I}\P_i(Z=i)\leq \max\left(B,\frac{\overline{K}_I}{\log(I+1)}\right).$$
\end{lemma}
We apply this lemma with ${\mathcal M}$ instead of $\{0,\ldots,I\}$, whose log-cardinal is larger than $T2^{j-1}$ up to an absolute constant. We take $a_j$ such that
$$\frac{Cc_2c_\psi^2na_j^2T^32^{j}(a^{-2}2^{-2j}\wedge 1)}{\log|\mathcal M|}\leq B,$$
which is satisfied if
\begin{equation}\label{cont3}
	a_j^2\leq \frac{C(\psi)}{nT^2c_2}(a^{2}2^{2j}\vee 1),
\end{equation}
with $C(\psi)$ a constant only depending on $\psi$ small enough. Thus 
if $Z$ is a random variable with values in
$m$, $\inf\limits_{m\in \M}\P_m(Z=m)\leq B.$
Now,
\begin{align}\label{eq:inf}
	\inf_{Z_n}\sup_{f_X\in {\mathcal S}^\beta(L,R)}
	\E_{f_X}\left[\|Z_n-f_X\|_{2,T}^2\right]&\geq\inf_{Z_n}\sup_{m\in \mathcal{M}}
	\E_{f_m}\left[\|Z_n-f_m\|_{2,T}^2\right]\nonumber\\
	&\geq\frac{1}{4}\inf_{m'\in \mathcal{M}}\sup_{m\in \mathcal{M}}
	\E_{f_m}\left[\|f_{m'}-f_m\|_{2,T}^2\right].
\end{align}
For the last inequality, we have used that if $Z_n$ is an estimate, we define  $$m'\in\arg\min_{m\in \mathcal{M}}  \E_{f_m}\left[\|Z_n-f_m\|_{2,T}^2\right]$$ and for $m\in \mathcal{M}$,
$$\|f_{m'}-f_m\|_{2,T}\leq \|f_{m'}-Z_n\|_{2,T}+ \|f_m-Z_n\|_{2,T}\leq 2\|f_m-Z_n\|_{2,T}.$$
Since $f_m-f_{m'}=a_j
\sum_{k\in m \Delta m'}b_k \psi_{jk}$ 
and $(\psi_{jk})$ is an orthonormal family, we have for $m\not=m'$,
\begin{equation}\label{mino2}
	\|f_m-f_{m'}\|_{2}^2= a_j^2 |m\Delta m'|\geq \tau a_j^2 T2^{j-1},
\end{equation}
for $\tau$ the absolute constant defined in Lemma~\ref{Gallager}. Furthermore,
\begin{align*}
	0\leq \|f_m-f_{m'}\|_2^2-\|f_m-f_{m'}\|_{2,T}^2&=\int(f_m(x)-f_{m'}(x))^21_{\{|x|>T\}}dx\\
	&=a_j^2\int\Big(\sum_{k\in m \Delta m'} b_k\psi_{jk}(x)\Big)^21_{\{|x|>T\}}dx\\
	&\leq a_j^2\sum_{k\in m \Delta m'} b_k^2\times\sum_{k\in m \Delta m'}\int\psi_{jk}^2(x)1_{\{|x|>T\}}dx\\
	&\leq a_j^2 |m\Delta m'|\times\sum_{k\in m \Delta m'}\int\psi_{jk}^2(x)1_{\{|x|>T\}}dx.
\end{align*}
Then, since $0\leq k\leq T2^{j-1}$,
\begin{align*}
	\int\psi_{jk}^2(x)1_{\{|x|>T\}}dx&=\int\psi^2(u)1_{\{|2^{-j}(u+k)|>T\}}du\\
	&\leq c_\psi^2\int(1+|u|)^{-4}1_{\{|2^{-j}(u+k)|>T\}}du\\
	&\leq  c_\psi^2\int_{2^jT-k}^{+\infty}(1+|u|)^{-4}du+ c_\psi^2\int_{-\infty}^{-2^jT-k}(1+|u|)^{-4}du\\
	&\leq  2c_\psi^2\int_{2^{j-1}T}^{+\infty}(1+u)^{-4}du\\
	&\leq \frac{2c_\psi^2}{3}\big(T2^{j-1}\big)^{-3}.
\end{align*}
We finally obtain
$$0\leq \|f_m-f_{m'}\|_2^2-\|f_m-f_{m'}\|_{2,T}^2\leq\frac{2c_\psi^2}{3}\big(T2^{j-1}\big)^{-1}\times a_j^2$$
and using \eqref{mino2}, for $j$ larger than a constant depending on $T$ and $c_\psi$, and for $m\not=m'$,
$$\|f_m-f_{m'}\|_{2,T}^2\geq C'(\psi)a_j^2 T2^j,$$
for $C'(\psi)$ a constant only depending on $\psi$. Finally, applying \eqref{eq:inf} and Lemma~\ref{fano}, we obtain:
\begin{align*}
	\inf_{Z_n}\sup_{f_X\in {\mathcal S}^\beta(L,R)}
	\E\left[\|Z_n-f_X\|_{2,T}^2\right]&\geq\frac{1}{4}\inf_{m'\in \mathcal{M}}\sup_{m\in \mathcal{M}}
	\E_{f_m}\left[\|f_{m'}-f_m\|_{2,T}^2\right]\\
	&\geq \frac{C'(\psi)a_j^2 T2^j}{4}\inf_{m'\in \mathcal{M}}\sup_{m\in \mathcal{M}}\P_m(m' \not=m)\geq \frac{C'(\psi)a_j^2 T2^j}{4}(1-B).
\end{align*}
Now, we choose $a_j>0$ as large as possible such that \eqref{cont1}, \eqref{cont2} and \eqref{cont3} are satisfied, meaning that
$$a_j^2T2^j=\left(C(\psi,\beta) L^22^{-2j\beta}\right)
\wedge 
\left(\frac{c_1^2T\bar C(T,c_\psi)^{-2}}{4}\right)
\wedge 
\left(\frac{C(\psi)2^j}{nTc_2}(a^{2}2^{2j}\vee 1)\right).
$$
Since $c_2=\max(4,2+4a^2)/c_1$ and $c_1=\mathfrak c_{\beta}^{-1} L/2$, it  simplifies in
$$a_j^2T2^j=\square_{T,\psi,\beta,a}\left(L^22^{-2j\beta}\wedge L^2 \wedge \frac{L2^{j}}{n}(2^{2j}\vee 1)\right).$$
We can take $j$ such that $T2^{j-1}\geq 1$ and
$$2^j\leq \left(Ln\right)^{\frac{1}{2\beta+3}}<2^{j+1}$$ 
for $n$ larger than a constant depending on $r$ and $T$ (since $L$ is larger than $2r \mathfrak c_\beta/\pi$),
which yields 
$$\inf_{Z_n}\sup_{f_X\in{\mathcal S}^\beta(L,r,b)}
\E\left[\|Z_n-f_X\|_{2,T}^2\right] \geq \square_{T,\psi,\beta,a} L^{\frac{2\beta+6}{2\beta+3}}n^{-\frac{2\beta}{2\beta+3}}.$$
Similarly, we can also take $j$ a constant depending on $T$ so that
$$\inf_{Z_n}\sup_{f_X\in {\mathcal S}^\beta(L,r,b)}
\E\left[\|Z_n-f_X\|_{2,T}^2\right] \geq \square_{T,\psi,\beta,a} Ln^{-1}.$$
This yields
$$\inf_{Z_n}\sup_{f_X\in {\mathcal S}^\beta(L,r,b)}
\E\left[\|Z_n-f_X\|_{2,T}^2\right] \geq \square_{T,\psi,\beta,a} \left[L^{\frac{2\beta+6}{2\beta+3}}n^{-\frac{2\beta}{2\beta+3}}+Ln^{-1}\right]$$
and Theorem~\ref{t2} is proved.


\subsection{Proof of Corollary~\ref{cor}}\label{sec:proofcor}
To prove Corollary~\ref{cor}, we combine the upper bound \eqref{UB} and the decomposition \eqref{BV} to obtain for any $f_X$ and any $h\in{\mathcal H}$,
\begin{align*}
	\E\left[\|\widetilde f-f_X\|_{2,T}^2\right]&\leq C_1 \E\left[\|\widetilde f_h-f_X\|_{2,T}^2\right]+\frac{C_2}{n}
	= C_1(B_h^2+v_h)+\frac{C_2}{n},
\end{align*}
where $C_1$ depends only on $\eta$ and $K$ and
$$B_h=\|K_h*f_X-f_X\|_{2,T},$$
$$v_h=\frac{aT\|f_X\|_1 \|K'\|_2^2}{2nh^3},$$
and 
$$C_2=(\|K\|_1+1)^2\left(c_1\|f_X\|_1^{-1}+ {c_2\|f_X\|_1}+c_3\right),$$
where $c_1$, $c_2$ and $c_3$ only depend on $\delta$, $a$, $K$, $T$ and $\eta$.
Assuming $f_X\in{\mathcal S}^\beta(L,r,b)$, we have
$$C_2\leq \square_{\delta, a, K, T, r,\eta,b}L.$$
Under Assumption~\ref{Hyp2}, we have for any $f_X\in{\mathcal S}^\beta(L,r,b)$
$$B_h=\|K_h*f_X-f_X\|_{2,T}\leq MLh^{\beta}$$
for $M$ a positive constant depending on $K$ and ${\beta}$. Indeed, the space ${\mathcal S}^\beta(L,r,b)$ is included into the Nikol'ski ball $H(\beta, L')$ with $L'$ equal to $L$ up to a constant. We refer the reader to Proposition~1.5 of \cite{tsybakov} and \cite{KLP2001} for more details. Now, we plug $h\in{\mathcal H}$ of order $(Ln)^{-\frac{1}{2\beta+3}}$ in the previous upper bound to obtain the desired bound of Corollary~\ref{cor} thanks to Lemma~\ref{lem:varsym}.

\section*{Acknowledgements}
The authors would like to thank the anonymous referee for constructive comments and suggestions leading to improvements of the paper. This work was supported by a grant from the Agence Nationale de la Recherche  ANR-18-CE45-0023 SingleStatOmics.
\bibliographystyle{apalike}

\end{document}